\documentclass[11pt]{article}
\usepackage{subcaption} 
\usepackage{color}
\usepackage{amsmath}
\usepackage{amssymb}
\usepackage{graphicx}
\usepackage{amsthm}
\usepackage{stmaryrd}
\usepackage[all]{xy}
\usepackage{multirow}
\usepackage{bm}
\def\lb{\ensuremath{\llbracket}}
\def\rb{\ensuremath{\rrbracket}}
\def\>{\ensuremath{\rangle}}
\def\<{\ensuremath{\langle}}

\newcommand {\supp } {{\rm supp}}

\newcommand {\E } {{\mathcal{E}}}
\newcommand {\F } {{\mathcal{F}}}

\newcommand {\D } {{\mathcal{D}}}
\renewcommand{\H}[0]{\mathcal{H}}

\newcommand {\M} {{\mathcal{M}}}

\newcommand {\tr} {{\mathrm{Tr}}}

\newcommand{\bra}[1]{\langle #1 \vert}
\newcommand{\nm}[1]{\lVert #1\rVert}
\newcommand{\ket}[1]{|#1\rangle}
\newcommand{\op}[2]{|#1\rangle \langle #2|}

\newtheorem{theorem}{Theorem}[section]

\newtheorem{lem}{Lemma}[section]
\newtheorem{defn}{Definition}[section]
\newtheorem{prop}{Proposition}[section]
\newtheorem{exam}{Example}[section]

\newtheorem{prob}{Problem}[section]

\newtheorem{fact}{Fact}[section]
\begin{document}

\title{Quantum Temporal Logic}

\author{Nengkun Yu \\University of Technology Sydney}

\maketitle

\begin{abstract}
In this paper, we introduce a model of quantum concurrent program, which can be used to model the behaviour of reactive quantum systems and to design quantum compilers.
We investigate quantum temporal logic, QTL, for the specification of quantum concurrent systems by suggesting the time-dependence of events.
QTL employs the projections on subspaces as atomic propositions, which was established in the Birkhoff and von Neumann's classic treatise on quantum logic.
For deterministic quantum program with exit, we prove a quantum B\"{o}hm-Jacopini theorem which states that any such program is equivalent to a Q-While program. The decidability of basic QTL formulae for general quantum concurrent program is studied.
\end{abstract}

\section{Introduction}\label{sec:intro}

The birth of science and logic are inextricably woven. The evolution of computer science and technology would have been impossible without its logical foundation. Conversely, the new model of computation has provided great opportunities for the development of logic. This ``entanglement'' between logic and computer science has lasted throughout the last century till now, even before the emergence of ENIAC--the first electronic general-purpose computer. Logic is ubiquitous in modern computer science. One typical example is temporal logic, the logical formulism for reasoning about time and the timing of events which was introduced to computer science by Pnueli in his seminal work \cite{4567924}.
Temporal logic has become a widely accepted language for stating and specifying properties of concurrent programs and general reactive systems \cite{PNUELI198145}. It is interpreted over models that abstract away from the actual time at which events occur, retaining only temporal ordering information about the states of a system.
The necessity of introducing temporal logic is from the following observations. The correctness of concurrent programs and general reactive systems usually involves reasoning about the corresponding events at different moments in an execution of the system \cite{Owicki:1982:PLP:357172.357178}. The expected property of reactive systems, e.g., the liveness properties--something good must eventually happen, as well as fairness properties--a fundamental concept of unbounded nondeterminism of concurrent systems \cite{FN86}, can not be stated as a formula of static logic.

Due to the increasing maturity in the research on quantum program verification, and the potential interest and understanding of the behavior of concurrency in quantum systems, great effort has been expended in the design of quantum programming language \cite{Kn96,SZ00,BCS03,Sa03,Se04,AG05,GLR13,JFD12,WS14,PRZ17,PZ19,O03,Google18,Rigetti18,GM19,LFBY17,JM05,FDY12,Wu}, as surveyed in \cite{Ga06,Se04a}. Simultaneously, various quantum logics and techniques have been developed for the verification of quantum programs
\cite{Ak05,BS04,BS06,BJ04,CMS,DP06,FDJY07,FYY13,GNP08,KA09,AGN14,RA17,SS15,Ying16,YFYY13,YLYF14,YYW17}. Notably,
Ying \cite{Ying11} established quantum Hoare
logic for both partial correctness and total correctness with (relative) completeness for the notion of quantum weakest precondition proposed by D'Hondt and Panangaden ~\cite{DP06}.
Of all the properties, the termination of quantum programs
has received a fair amount of attention \cite{JFD12,YF10,YYFD13,LYY14,YY12,LY17}.
Quantum process algebra was introduced to model the quantum communication between quantum processors, and thus forms a model of concurrent quantum systems \cite{JM05,YFDJ09,FDY12}. Temporal logics for  quantum states were introduced and then the model-checking problem for this logic was studied in \cite{MS06,BCMS07,BCM08,MRSS09}.

Concurrency is a necessary concept to explore the computational power of distributed quantum computing systems.
Most of the aforementioned research on the correctness of quantum programs considered programs with exits only. Those are programs with a distinct beginning and end with computational instruction in between, whose correctness statement describes the relation between the input and successful completion of the program.
These approaches completely
ignore an important class of quantum operating systems or
real-time type quantum programs, for instance, the quantum internet \cite{YLZ19}, for which halting is
an abnormal situation.

Understanding the combined behaviour of quantum features and concurrency is nontrivial. Quantum mechanics naturally generate probabilistic branching due to quantum measurement as does quantum programs. On the other hand, non-determinism plays a central role in concurrent programs. A comprehensive definition of a quantum concurrent program which would intergrade quantum probability and non-determinism is still missing.
Previous work including \cite{YY12} should be regarded as concurrent quantum programs but not quantum concurrent programs because that concurrency does not depend on quantum data but purely classical.

To define the propositional variables of quantum temporal logic, we revisit the (algebraic counter-part of) quantum logic of quantum mechanics which originated in the milestone paper \cite{BvN36} by Garrett Birkhoff and John von Neumann, who were attempting to reconcile the apparent inconsistency of classical logic with the facts concerning the measurement of complementary variables in quantum mechanics, at which time quantum logic was defined as the set of principles for manipulating the projections on a Hilbert space which was viewed as quantum propositions about physical observables in John von Neumann's classic treatise \cite{von18}. Projective operators in a Hilbert space correspond one-to-one with the closed subspaces, and the L$\ddot{\mathrm{o}}$wner order restricted on projection operators coincides with the inclusion between the corresponding subspaces. The structure of the set of closed subspaces of a Hilbert has been thoroughly investigated in the development of quantum logic for over 80 years.
The idea of using quantum projections as a quantum predicate was discussed in \cite{YDFJ10}, where the algebraic structure of the set of closed subspaces and orthomodular lattices theory is reviewed \cite{BH00,KA83}.  Recently, this idea was also used to develop quantum relational Hoare logic (qRHL) in \cite{Unruh,Unruh2,BHY19}. It was employed in providing an applied quantum Hoare Logic in \cite{ZYY19}.

\textbf{Contributions of the Paper}: In this paper, we derive a \textit{quantum Temporal logic} (QTL) as a verification tool of quantum concurrent programs. More precisely, our contribution is as follows:
\begin{itemize}
\item
 We provide a model of quantum concurrent programs which combines quantum probability and non-determinism. In this model, the quantum concurrent program consists of a shared quantum register, a class of quantum programs which can access the quantum register and a scheduler (classical) register which records the program that needs to be performed in the next round. Each program is given as finite lines (locations) of commands with an initial location. Each command consists of a quantum super-operator and a measurement in which the classical index outcome of measurements is used to choose a location of this program and modify the scheduler register, which can be non-deterministically.

\item We investigate quantum temporal logic, QTL, which generalizes Pnueli's classical temporal logic \cite{4567924}. Birkhoff and von Neumann used projection as quantum atomic propositions where a state satisfies a proposition if the state falls into the subspace corresponding to the projection in \cite{BvN36}. In light of this method, we define the basic temporal operators, $\mathbf{O}$ (next), $\mathbf{U}$ (until), $\tilde{\mathbf{U}}$ (almost surely until) $\mathbf{true}$, $\lozenge$ (eventually), $\tilde{\lozenge}$ (almost surely eventually) and $\square$ (always).
    Note that we do not allow negation although $\wedge$ (conjunction), $\vee$ (disjunction) are introduced as usual.

\item We study the QTL for deterministic quantum concurrent programs with exits as an example of the general model. We provide a quantum compiler for Q-While, a widely studied quantum extension of the while-language \cite{Ying11}. We prove a quantum B\"{o}hm-Jacopini theorem \cite{BJ66} which states that any deterministic quantum concurrent program is equivalent to a Q-While program. Based on this theorem, for deterministic quantum concurrent program,
    \begin{itemize}
\item
    (a) we present a logic with completeness for reasoning and thus fill an important gap in the verification of quantum programs;
    \item (b) we provide polynomial time algorithms which compute the reachability super-operator and average running time;
    \item (c) we demonstrate a quantum analogue of the Kleene closure which compute the entanglement-assisted reachable space.
\end{itemize}
\item We study the decidability of basic QTL formulae. For deterministic quantum program, we show that $\square\tilde{\mathbf{U}}$ is decidable while the decidabilities of $\lozenge$, $\tilde{\lozenge}$,$\mathbf{U}$ and $\tilde{\mathbf{U}}$ are equivalent to the decidability of the famous Skolem problem. For general quantum concurrent programs, we prove that $\square$, $\square\lozenge$, $\lozenge\square$ and $\square\mathbf{U}$ are all decidable which solves the open question of \cite{LY14}.
\end{itemize}
{\vskip 3pt}
We list the reasons for employing projection as quantum atomic propositions in the following.
\begin{itemize}
\item It enables us to define logical operators $\wedge$ (conjunction), $\vee$ (disjunction).

\item Each projection $P$ corresponds to a projective measurement $\{P,I-P\}$ which is physically implementable. The state satisfies $P$ if and only if the measurement outcomes $P$ when applying $\{P,I-P\}$ on the state. Moreover, the state, if satisfies $P$, will not collapse after applying the measurement. Therefore, our logic fits very well in the testing and debugging of quantum concurrent programs.

\item For deterministic quantum program with exit, the set of input states such that the program terminates in finite steps, and the set of input states such that the program terminates with probability 1, can be characterized by closed subspaces, respectively, or equivalently projections, as observed in \cite{ZYY19}.
\end{itemize}

\subsection{Related Work and Comparison}
\cite{YF11} defined a flowchart low-level quantum programming languages and provided a technique of
translating quantum flowchart programs into Q-While.

\cite{YING201831} introduced a quantum Markov decision process is as a semantic model of non-deterministic and quantum concurrent programs in which each program is given as a quantum operation and a finite set of measurements is given. At each step, a quantum program is applied or a measurement is performed.

Compared with \cite{YF11} and \cite{YING201831}, the classical control of the model here has a richer structure. Each program consists of a class of commands, marked in corresponding program locations, where each command is a tuple of a quantum operation together with a quantum measurement.
The classical control information is recorded in the scheduler register together with the locations of each program.
At each step, according to the value of the schedule register, the command corresponding to the current location of the corresponding program is applied. Our quantum B\"{o}hm-Jacopini theorem is stronger than the one of  \cite{YF11} in the sense that we only need a single syntax of Q-While to characterize the original deterministic program.

\textbf{Organisation of the Paper}: We provide preliminaries about quantum information in Section \ref{sec:preliminaries}. In Section \ref{sec:model}, we introduce the model of quantum concurrent programs. In Section \ref{sec:QTL}, we give the formal definition of quantum temporal logic (syntax, semantics). In Section \ref{sec:example}, we study the deterministic quantum concurrent programs with exits as an example.
In Section \ref{sec:complexity}, we studied the decidability of quantum temporal logic.

\section{Quantum Information: Preliminaries and Notations} \label{sec:preliminaries}
This section presents the background and notations on quantum information
and quantum computation mainly according to the textbook by~\cite{NI11}.

\subsection{Preliminaries}
A Hilbert space $\H$ is a linear vector space which can be finite dimensional or separable.
A separable Hilbert space has a countable orthonormal basis.
For any finite integer $n$, an $n$-dimensional Hilbert space $\H$
is the space $\mathbb{C}^n$ of complex vectors.
We use Dirac's notation, $\ket{\psi}$, to denote a complex vector in $\mathbb{C}^n$.
The inner product of two vectors $\ket{\psi}$ and $\ket{\phi}$ is denoted by $\langle\psi|\phi\rangle$,
which is the product of the Hermitian conjugate of $\ket{\psi}$, denoted by $\bra{\psi}$, and vector $\ket{\phi}$.
The norm of a vector $\ket{\psi}$ is denoted by $\nm{\ket{\psi}}=\sqrt{\langle\psi|\psi\rangle}$.

Linear \emph{operators} are linear mappings between Hilbert spaces.
Operators between $n$-dimensional Hilbert spaces are represented by $n\times n$ matrices.
For example, the identity operator $I_\H$ is the identity matrix on $\H$.
The Hermitian conjugate of operator $A$ is denoted by $A^\dag$. Operator $A$ is \emph{Hermitian} if $A=A^\dag$.
The trace of an operator $A$ 
is the sum of the entries on the main diagonal, i.e., $\tr(A)=\sum_i A_{ii}$. 
We write $\bra{\psi}A\ket{\psi}$ to mean the inner product between
$\ket{\psi}$ and $A\ket{\psi}$.
A Hermitian operator $A$ is \emph{positive semidefinite} (resp.,
\emph{positive definite}) if for all vectors $\ket{\psi}\in\H$,
$\bra{\psi}A\ket{\psi}\geq 0$ (resp., $>0$).
This gives rise to the \emph{L\"owner order} $\sqsubseteq$ among operators:
\begin{equation}
 A\sqsubseteq B \text{ if } B-A \text{ is positive semidefinite, } \quad A\sqsubset B \text{ if } B-A \text{ is positive definite. }
\end{equation}
A positive semidefinite operator $P$ is called a \emph{projection} if
\begin{equation}
P=P^{\dag}=P^2.
\end{equation}
There is a one-to-one correspondence between projection and closed linear subspace. the \emph{L\"owner order} $\sqsubseteq$ among projections is equivalent to
the subset relation among closed linear subspaces.
\subsection{Quantum States}

The state space of a quantum system is a Hilbert space.
The state space of a \emph{qubit}, or quantum bit, is a 2-dimensional Hilbert space.
One important orthonormal basis of a qubit system is the \emph{computational} basis with $\ket{0}=(1,0)^\dag$ and $\ket{1}=(0,1)^\dag$, which encode the classical bits 0 and 1 respectively.
Another important basis, called the $\pm$ basis, consists of $\ket{+}=\frac{1}{\sqrt{2}}(\ket{0}+\ket{1})$ and $\ket{-}=\frac{1}{\sqrt{2}}(\ket{0}-\ket{1})$.
The state space of multiple qubits is the \emph{tensor product} of single qubit state spaces.
For example, classical 00 can be encoded by $\ket{0}\otimes\ket{0}$
(written $\ket{0}\ket{0}$ or even $\ket{00}$ for short) in the Hilbert space $\mathbb{C}^2\otimes\mathbb{C}^2$.
The Hilbert space for an $m$-qubit system is $(\mathbb{C}^2)^{\otimes m} \cong \mathbb{C}^{2^m}$.

A \emph{pure} quantum state is represented by a unit vector, i.e., a vector $\ket{\psi}$ with $\nm{\ket{\psi}}=1$.
A \emph{mixed} state can be represented by a classical distribution over an ensemble of pure states $\{(p_i,\ket{\psi_i})\}_i$,
i.e., the system is in state $\ket{\psi_i}$ with probability $p_i$.
One can also use \emph{density operators} to represent both pure and mixed quantum states.
A density operator $\rho$ for a mixed state representing the ensemble $\{(p_i,\ket{\psi_i})\}_i$ is a positive semidefinite operator $\rho=\sum_i p_i\ket{\psi_i}\bra{\psi_i}$, where $\ket{\psi_i}\bra{\psi_i}$ is the outer-product of $\ket{\psi_i}$; in particular, a pure state $\ket{\psi}$ can be identified with the density operator $\rho=\ket{\psi}\bra{\psi}$.
Note that $\tr(\rho)=1$ holds for all density operators. A positive semidefinite operator $\rho$ on $\H$ is said to be a \emph{partial} density operator if $\tr(\rho)\leq 1$.
The set of partial density operators is denoted by $\D(\H)$.

\subsection{Quantum Operations}\label{sec:QTL}

Operations on closed quantum systems can be characterized by unitary operators. An operator $U$ is \emph{unitary} if its Hermitian conjugate is its own inverse, i.e., $U^\dag U=UU^\dag=I_{\H}$. For a pure state $\ket{\psi}$, a unitary operator describes an \emph{evolution} from $\ket{\psi}$ to $U\ket{\psi}$. For a density operator $\rho$, the corresponding evolution is $\rho \mapsto U\rho U^\dag$.
The \emph{Hadamard} operator $H$ transforms between the computational and the $\pm$ basis. For example, $H\ket{0}=\ket{+}$ and $H\ket{1}=\ket{-}$.

More generally, the evolution of a quantum system can be characterized by an \emph{super-operator} $\E$, which is a \emph{completely-positive} and \emph{trace-non-increasing} linear map from $\D(\H)$ to $\D(\H')$ for Hilbert spaces $\H, \H'$. For every super-operator $\E : \D(\H)\to\D(\H')$, there exists a set of Kraus operators $\{E_k\}_k$ such that $\E(\rho)=\sum_k E_k\rho E_k^\dag$ for any input $\rho\in\D(\H)$.
Note that the set of Kraus operators is finite if the Hilbert space is finite-dimensional.
The \emph{Kraus form} of $\E$ is written as $\E(\cdot)=\sum_k E_k\cdot E_k^\dag$.
A unitary evolution can be represented by the super-operator $\E(\cdot)=U\cdot U^\dag$. An identity operation refers to the super-operator $\mathcal{I}_{\H} (\cdot)= I_{\H} \cdot I_{\H}$.
A super-operator $\E$ is trace-non-increasing if for any initial state $\rho\in\D(\H)$, the final state $\E(\rho)\in \D(\H')$ after applying $\E$ satisfies $\tr(\E(\rho))\leq\tr(\rho)$. A super-operator $\E$ is called trace preserving if $\tr(\E(\rho))=\tr(\rho)$ holds for any initial state $\rho\in\D(\H)$.
The Schr\"odinger-Heisenberg \emph{dual} of a super-operator $\E=\sum_k E_k\circ E_k^\dag$, denoted by $\E^*$, is defined as follows: for every state $\rho\in\D(\H)$ and any operator $A$, $\tr(A\E(\rho))=\tr(\E^*(A)\rho)$. The Kraus form of $\E^*$ is $\sum_k E_k^\dag\cdot E_k$.

\subsection{Quantum Measurements}

The way to extract information about a quantum system is called
a quantum \emph{measurement}.
A quantum measurement on a system over Hilbert space $\H$ can be
described by a set of linear operators
$\{M_m\}_m$ with $\sum_m M_m^\dag M_m=I_\H$.
If we perform a measurement $\{M_m\}$ on a state $\rho$, the outcome $m$ is observed with probability $p_m=\tr(M_m\rho M_m^\dag)$ for each $m$.
A major difference between classical and quantum computation is that a
quantum measurement changes the state. In particular, after a
measurement yielding outcome $m$, the state collapses to $M_m\rho M_m^\dag/p_m$.
For example, a measurement in the computational basis is described by $M=\{M_0=\ket{0}\bra{0}, M_1=\ket{1}\bra{1}\}$.
If we perform the computational basis measurement $M$ on state $\rho=\ket{+}\bra{+}$, then with probability $\frac{1}{2}$ the outcome is $0$ and $\rho$ becomes $\ket{0}\bra{0}$.
With probability $\frac{1}{2}$ the outcome is $1$ and $\rho$ becomes $\ket{1}\bra{1}$. Quantum measurements are essentially probabilistic; but we adopt a convention from \cite{Se04} to present them. We can combine probability $p_m$ and density operator $\rho_m$ into a partial density operator $M_m\rho M_m^\dag=p_m\rho_m.$ This convention significantly simplifies the presentation.

A projective measurement on a system with state space $\mathcal{H}$ is
described by a collection $\{P_m\}$ of projections over $\mathcal{H}$
satisfying
$\sum_{m}P_m=I_{\mathcal{H}},$
where index $m$ stands for the measurement outcomes that may occur. If the state of a quantum system was $\rho$
immediately before the measurement is performed on it, then
the probability that outcome $m$ occurs is
$p_m=\tr(P_m\rho),$ and the state of the system after
the measurement is $\rho_m=P_m\rho P_m^{\dag}/p_m.$
Actually, a general measurement can always be implemented by a projective
measurement together with a unitary transformation if an ancillary system is allowed. In the circuit model of quantum computation, measurements are usually assumed to be in the computational basis, which is a special
kind of projective measurement.

For a mixed state (density operator) $\rho$, its support $\supp(\rho)$ is defined as the (topological) closure of the subspace spanned by the eigenvectors of $\rho$ with nonzero eigenvalues. It is easy to see that $\supp(\rho)=\{\ket{\varphi}\in\mathcal{H}:\langle\varphi|\rho\ket{\varphi}=0\}^{\perp}$,
where ${}^\perp$ stands for ortho-complement. The definition of support can be naturally generalized to semi-definite positive operators. An important fact of projective measurements is that, given a state $\rho$ and projection $P$ such
that $\supp(\rho)\subseteq P$, if we apply the (yes/no) projective measurement
$\{P,I-P\}$ on $\rho$, the state is not changed.
\subsection{Jordan Canonical Forms}
Let $M\in \mathbb{Q}^{d\times d}$ be a square matrix with rational entries. The minimal polynomial of $M$ is the unique monic polynomial $m(x)\in \mathbb{Q}[x]$ of least degree such that $m(A)=0$. By the Cayley-Hamilton Theorem, the degree of $m(x)$ is at most $d$.

We can write any matrix $M\in\mathbb{C}^{d\times d}$ as $M=P^{-1}J P$ for some invertible matrix $P$ and block diagonal Jordan matrix $J=diag(J_1,...,J_N)$, with each block $J_i$ with size $l\times l$ having the following form
\begin{equation}
J_i=\begin{bmatrix}
    \lambda_i & 1 & 0 & \dots  & 0 \\
    0 & \lambda_i & 1 & \dots  & 0 \\
    \vdots & \vdots & \vdots & \ddots & \vdots \\
    0 & 0 & 0 & \dots  & \lambda_i
\end{bmatrix}
\Rightarrow J_i^m=\begin{bmatrix}
    \lambda_i^m & m\lambda_i^{m-1}  & {m\choose 2}\lambda_i^{m-2} & \dots  & {m\choose {l-1}}\lambda_i^{m-l+1} \\
    0 & \lambda_i^m & m\lambda_i^{m-1} & \dots  & 0 \\
    \vdots & \vdots & \vdots & \ddots & \vdots \\
    0 & 0 & 0 & \dots  & \lambda_i^m
\end{bmatrix}
\end{equation}
where the binomial coefficient $m\choose j$ is defined to be 0 for $m<j$.

Moreover, given a rational matrix $M\in \mathbb{Q}^{d\times d}$, its Jordan Normal Form $M=P^{-1}J P$ can be computed in polynomial time, as shown in \cite{CAI94}. Here, the input size of the problem is the total lengths of the binary representation of all the input entries, and the complexity is measured in terms of binary operations. We associate each algebraic number with its minimal polynomial (thus, irreducible) and a sufficiently good rational approximation, which uniquely identifies the particular root of the polynomial.
\subsection{Matrix Representation of Super-Operators}\label{MRSO}

The matrix representation of a super-operator is usually easier to
manipulate than the super-operator itself.
\begin{defn}\label{mrd}
Suppose super-operator $\mathcal{E}$ on a finite-dimensional Hilbert
space $\H$ has the operator-sum representation
$\mathcal{E}(\rho)=\sum_iE_i\rho E_i^{\dag}$ for all partial density
operators $\rho$, and $\dim \H=d$. Then the matrix representation of
$\mathcal{E}$ is the following $d^2\times d^2$ matrix: $$M=\sum_i
E_i\otimes E_i^*,$$ where $A^{\ast}$ stands for the conjugate of
matrix $A$, i.e. $A^{\ast}=(a^{\ast}_{ij})$ with $a^{\ast}_{ij}$
being the conjugate of complex number $a_{ij}$, whenever
$A=(a_{ij})$.
\end{defn}

The following lemma illustrates the usefulness of the matrix representation of super-operator \cite{YYFD13}.
\begin{lem}\label{mr-lem} We write
$|\Phi\rangle=\sum_j|jj\rangle$ for the (unnormalized) maximally
entangled state in $\H\otimes \H$, where $\{|j\rangle\}$ is an
orthonormal basis of $\H$. Let $M$ be the matrix representation of
super-operator $\mathcal{E}$. Then for any $d\times d$ matrix $A$,
we have: $$(\mathcal{E}(A)\otimes I)|\Phi\rangle=M(A\otimes
I)|\Phi\rangle.$$\end{lem}
Let the matrix representations of super-operators
$\mathcal{E}$ be $M$ with Jordan decomposition $M=SJ(M)S^{-1},$
where $S$ is a nonsingular matrix, and $J(M)$ is the Jordan normal
form of $M$: $$J(M)= diag (J_{k_1}(\lambda_1),
J_{k_2}(\lambda_2),\cdot\cdot\cdot, J_{k_l}(\lambda_l))$$ with
$J_{k_s}(\lambda_s)$ being a $k_s\times k_s$-Jordan block of
eigenvalue $\lambda_s$ $(1\leq s\leq l)$. The next lemma describes the structure of the matrix
representation $M$ of super-operator $\mathcal{F}$.

\begin{lem}\label{tech0}\begin{enumerate}\item $|\lambda_s|\leq 1$ for all $1\leq s\leq l$. \item If $|\lambda_s|=1$
then the dimension of the $s$th Jordan block $k_s=1$.\end{enumerate}
\end{lem}
\subsection{Convergence of the decreasing chain of finite union of subspaces}
We use the following result proved in \cite{LYY14}.
\begin{lem}\label{uss}
Suppose $X_k$ is a union of a finite number of subspaces of $\H$ for all $k\geq 1$. If $X_k$ is a decreasing chain, i.e., $X_1\supseteq X_2\supseteq \cdots \supseteq X_k\supseteq$, then there exists $n\geq 1$ such that $X_k=X_n$ for all $k\geq n$.
\end{lem}
\subsection{Kronecker's Theorem}
The classical Kronecker approximation theorem is formulated as follows.
\begin{theorem} \label{K1}
Given real n-tuples $\alpha=(\alpha_1,\cdots,\alpha_n)\in \mathbb{R}^n$ and $\beta=(\beta_1,\cdots,\beta_n)\in \mathbb{R}^n$, the condition:
$\forall \epsilon>0$, $\exists p,q_j\in\mathbb{N}$ such that $$|p\alpha_i-q_j-\beta_i|<\epsilon,\forall 1\leq j\leq n$$
    holds if and only if for any $r_1,\cdots,r_n\in\mathbb{Z}$ with $\sum_{j=1}^n r_j\alpha_j\in\mathbb{Z}$, $\sum_{j=1}^n r_j\beta_j$ is also an integer.
\end{theorem}
In simpler language, the first condition states that the tuple $\beta=(\beta_1,\cdots,\beta_n)\in \mathbb{R}^n$ can be approximated arbitrarily well by integer scaling  of $\alpha$ and integer vectors. In other words, the decimal part of $n\alpha$ is dense in the set $\{(\beta_1,\cdots,\beta_n)|\sum_{j=1}^n r_j\beta_j\in\mathbb{Z},~~\forall r_j\in\mathbb{Z}~s.t.~\sum_{j=1}^n r_j\alpha_j\in\mathbb{Z}\}$.

To characterize the limit points of $n\alpha$, we also need the following result from \cite{GE93} which characterized all $(r_1,\cdots,r_n)\in\mathbb{N}^n$ such that $\sum_{j=1}^n r_j\alpha_j\in\mathbb{Z}$, or equivalently $\Pi_{j=1}^n \exp(i2\pi r_j\alpha_j)=1$.
\begin{theorem} \label{K2}
Given algebraic numbers, $\lambda_1,\cdots,\lambda_n$, one can in polynomial time calculate the following set
$$
La:=\{(k_1,\cdots,k_n)|\Pi_{j=1}^n \lambda_j^{k_j}=1\}\subset\mathbb{Z}^n.
$$
\end{theorem}
Remark: $La$ is a lattice, i.e., $av_1+bv_2\in La$ for all $v_1,v_2\in La$ and $a,b\in\mathbb{Z}$. The algorithm outputs a basis of the lattice $La$.

\subsection{Skolem-Mahler-Lech Theorem and Skolem Problem}
The Skolem-Mahler-Lech theorem is useful in our analysis. We use the version in \cite{HAN86}.
\begin{theorem} \label{SML}
If the zero set of a linear recurrence series $a_n=\tr(\op{u}{v}A^n)$ with $\ket{u},\ket{v}$ and $A$ being $d$ dimensional integer vectors and invertible matrix is infinite, it is eventually periodic, i.e. it agrees with a periodic set for sufficiently large $n$. In fact, a slightly stronger statement is true: the zero set is the union of a finite set and a finite number of residue classes $\{ n \in {\Bbb N}: n = k \mod r \}$. For any prime number $p\nmid 2\det(A)$, $r$ can be bounded by $r\leq p^{d^2}$.
\end{theorem}
It is not known whether the following Skolem Problem is decidable for $d\geq 5$.
\begin{prob}\label{Skolem}
Given a linear recurrence set $a_n=\tr(\op{u}{v}A^n)$ with $\ket{u},\ket{v}$ and $A$ being $d$ dimensional integer vectors and invertible matrix, determine whether the zero set of $a_n$ is empty.
\end{prob}

\section{Systems and Programs}\label{sec:model}

Before providing the general model of quantum sequential and quantum concurrent programs, we recall the framework of classical systems provided in \cite{KEL76}.
\begin{defn}
A dynamic discrete system consists of
$$
<S,R,s_0>
$$
where:
\begin{itemize}
\item $S$ is the set of states the system may assume (possibly infinite).
\item $R$ is the transition relation holding between a state and its possible successors, $R\subset S\times S$.
\item $s_0$ is the initial state $s_0\in S$.
\end{itemize}
An execution of the system is a sequence:
$$
\mathfrak{S}=s_0s_1\cdots s_i\cdots
$$
where for each $i\geq 0$, $R(s_i,s_{i+1})$ holds.
The system is deterministic if for any $s\in S$, there is only one $t\in S$ such that $R(s,t)$ holds. Otherwise, it is non-deterministic.
\end{defn}
Similarly, we can form our quantum system as follows: A dynamic discrete time quantum system consists of
$$
<\mathcal{H},R,\rho_0>
$$
where:
\begin{itemize}
\item $\mathcal{H}$ is the Hilbert space that the system may assume (finite dimensional or separable).
\item $R$ is the transition relation holding between a state and its possible successors, $R\subset \mathcal{D}(\mathcal{H})\times \mathcal{D}(\mathcal{H})$.
\item $\rho_0$ is the initial state with $\rho_0\in\mathcal{D}(\mathcal{H})$.
\end{itemize}
An execution of the system is a sequence:
$$
\mathfrak{S}=\rho_0\rho_1\cdots\rho_i\cdots
$$
where for each $i\geq 0$, $R(\rho_i,\rho_{i+1})$ holds.
Many different execution sequences are possible as $R$ is nondeterministic in general.

As quantum mechanics is linear, the most natural choice of $R$ is a super-operator introduced in Section \ref{sec:preliminaries}.
However, this would not result in non-determinism. Non-determinism always corresponds to a finite discrete set. Like the classical framework, we would like the non-determinism to depend on the state. This motivates us to introduced measurement at each step of transition because quantum measurement is the only way to extract classical information from quantum systems.

This concept of discrete quantum system discussed below is very general. Being chiefly motivated by problems in the quantum programming area, all the examples and following discussions will be addressed to the verification of programs. Further structuring of state notion is needed to particularize quantum system into quantum programs.

\subsection{Sequential Quantum Programs}
The sequential quantum programs is given in the following structure.
\begin{defn}
A sequential quantum program is a six tuple
$$
\pi=(\H,L,Act,Q,\rho_0,l),
$$
where
\begin{itemize}
\item $\H$ is a Hilbert space, called the state space. $\H$ contains the data component and ranges over an infinite domain, the quantum state of $\H$. It can be freely structured into individual variables and data structures for fitting actual applications.
\item $L$ is the control component and assumes a finite number of values, taken to be labels or locations in the program. Without loss of generality, we let $L=\{l_0,l_1,\cdots,l_n\}$ be the set of locations, $|L|$ can be regarded as the program length.
\item $Act$ is a mapping which associates each location $l_i\in L$ with a corresponding trace preserving super-operator $\E_{l_i}:\D(\H)\mapsto\D(\H)$ and a quantum measurement $\M_{l_i}=\{M_{l_i,0},\cdots,M_{l_i,N}\}$ satisfying $M_{l_i,j}:\H\mapsto\H$. They are used to describe the evolution of the system caused by action.  Note that we can assume a uniform $N$ as the number of outcomes for the measurements at all locations because we assume $L$ is finite.
\item $Q:\{0,1,\cdots,N\}\times L\mapsto 2^{\{l_0,l_1,\cdots,l_n\}}\setminus\emptyset$ denotes the next location choice mapping.
\item $\rho_0$ is the initial state of the system, which lies in $\H$.
\item $l_0\in L$ is the initial location of the program.
\end{itemize}
where $L$ can be regarded as the control component of the program.
\end{defn}
To clarify the transition of the system, we look at the joint distribution of the locations and the quantum data. Due to the probability distribution induced by quantum measurements, the actual state including the location of the program, is not always of the form $\rho\otimes \op{l_i}{l_i}$, but is as follows:
$$
\{\sum_{i=0}^n \rho_i\otimes \op{l_i}{l_i}|\rho_i\geq 0, \sum_{i=0}^n\tr(\rho_i)=1\}
$$
where $\rho_i\in\D(\H)$.

We express the transition as follows.

In the first step, $\E_{l_0}$ is applied on the initial state
$$
\rho_0\otimes\op{l_0}{l_0}\mapsto \E_{l_0}(\rho_0)\otimes\op{l_0}{l_0}.
$$
Then, measurement $\M_{l_0}$ is performed and based on the measurement outcome, classical index, and the corresponding location, the
location is changed accordingly. In other words, any state of the following form with $f(j,l_0)\in Q(j,l_0)$ is reachable non-deterministically
$$
\sum_{j=0}^N M_{l_0,j}\E_{l_i}(\rho_i)M_{l_0,j}^{\dag}\otimes \op{f(j,l_0)}{f(j,l_0)}.
$$
Generally, at each step, the state $\sum_{i=0}^n \rho_i\otimes \op{l_i}{l_i}$ is transformed by the following two sub-steps.
\begin{itemize}
\item Apply super-operators according to the location and obtain $\sum_{i=0}^n \E_{l_i}(\rho_i)\otimes \op{l_i}{l_i}$.
\item Apply quantum measurement and change location accordingly.  The overall state becomes the following for any $f(j,l_i)\in Q(j,l_i)$ non-deterministically
$$\sum_{i=0}^n\sum_{j=0}^N M_{l_i,j}\E_{l_i}(\rho_i)M_{l_i,j}^{\dag}\otimes \op{f(j,l_i)}{f(j,l_i)}.$$
\end{itemize}

\subsubsection{Quantum Sequential Program with Exit}

This program becomes deterministic if and only if $|Q(j,l_i)|=1$ for all $0\leq j\leq N$ and $l_i\in L$.

Our model of sequential quantum program can also simulate quantum programs with exit by assuming an exit location $l_e\in L$ such that $\E_{l_e}$ and $\M_{l_e}$ do not change the state, and $Q(j,l_e)=\{l_e\}$ for all $0\leq j\leq N$. More precisely, we let $\E_{l_e}(\rho)=\rho$ for all $\rho\in\D(\H)$ and $\M_{l_e}=\{I_{\H},0,0,\cdots,0\}$.

We define two kinds of terminations of quantum programs with exits $\pi$ with exit location $l_e$.
\begin{defn}
Let $\sigma_0=\rho_0\otimes\op{l_0}{l_0}$, and $\sigma_k$ denote the state of the system at step (time) $k$.
\begin{itemize}
\item We say that $\pi$ terminates if there exists $n$ such that $\tr[\sigma_n(I_{\H}\otimes\op{l_e}{l_e})]=1$.
\item We say that $\pi$ almost terminates if for any $\delta>0$ there exists $n$ such that $\tr[\sigma_k(I_{\H}\otimes\op{l_e}{l_e})]>1-\delta$ for all $k>n$.
\end{itemize}
\end{defn}
The termination of a quantum program is rare whereas almost terminaition is much more common as illustrated in the following example.

\begin{exam}\label{example1}
\begin{enumerate}
\item  $l_1$: $\rho_0=\op{-}{-}$, goto $l_2$;
 \item $l_2$: Measure $\rho$ using $\{M_0=\ket{0}\bra{0}, M_1=\ket{1}\bra{1}\}$, if the outcome is $0$, goto $l_4$, otherwise, goto $l_3$;
 \item $l_3$: Apply $H$ gate on $\rho$, goto $l_2$
 \item $l_4$: goto $l_4;$
\end{enumerate}
\end{exam}
This program $\pi$ is a deterministic program with exit, with $l_4$ as the exit location.
\begin{itemize}
\item The initial state is $\sigma_0=\op{-}{-}\otimes\op{l_1}{l_1}$.
\item After the first step, the state becomes $\sigma_1=\op{-}{-}\otimes\op{l_2}{l_2}$.
\item After the second step, the state becomes $\sigma_2=\frac{1}{2}\op{0}{0}\otimes\op{l_4}{l_4}+\frac{1}{2}\op{1}{1}\otimes\op{l_3}{l_3}$.
\item After the third step, the state becomes $\sigma_3=\frac{1}{2}\op{0}{0}\otimes\op{l_4}{l_4}+\frac{1}{2}\op{-}{-}\otimes\op{l_2}{l_2}$.
\item After the fourth step, the state becomes $\sigma_4=\frac{3}{4}\op{0}{0}\otimes\op{l_4}{l_4}+\frac{1}{4}\op{1}{1}\otimes\op{l_3}{l_3}$.
\item $\cdots$
\item After the $2n$-th step, the state becomes $\sigma_{2n}=(1-\frac{1}{2^n})\op{0}{0}\otimes\op{l_4}{l_4}+\frac{1}{2^n}\op{1}{1}\otimes\op{l_3}{l_3}$.
\item After the $2n+1$-th step, the state becomes $\sigma_{2n+1}=(1-\frac{1}{2^n})\op{0}{0}\otimes\op{l_4}{l_4}+\frac{1}{2^n}\op{-}{-}\otimes\op{l_4}{l_4}$.
\item $\cdots$
\end{itemize}
For any finite $n$, the program will not reach $\op{l_4}{l_4}$ exactly. That is, $\pi$ does not terminate.
On the other hand, $\pi$ almost terminates because for any $\delta>0$, we can find $n$ such that for any $k>n$ such that $\tr[\sigma_n(I_{\H}\otimes\op{l_e}{l_e})]>1-\delta$.
\subsection{Quantum Concurrent Programs}
To illustrate the quantum concurrent programs, we allow more than one control component. The following model implements the process call, although it does not behave like the recursive model of the quantum programs with exits investigated in \cite{FYY13b}.

\begin{defn}
An $m$-party quantum concurrent program is an eight tuple
$$
\pi=(\H,L,Act,S,Q,\rho_0,l,s),
$$
where
\begin{itemize}
\item $\H$ is a Hilbert space, called the state space.
\item $L=(L_1,L_2,\cdots,L_m)$ with $L_i=\{l_{i,0},l_{i,1},\cdots,l_{i,m_i}\}$ being the control component, locations, of process $P_i$.
\item $Act$ is a mapping which associates each location $l_{i,j}\in L_i$ with a corresponding trace preserving super-operator $\E_{l_{i,j}}:\D(\H)\mapsto\D(\H)$ and a quantum measurement $\M_{l_{i,j}}=\{M_{l_{i,j},0},\cdots,M_{l_{i,j},N}\}$ satisfying $M_{l_{i,j},k}:\H\mapsto\H$. They are used to describe the evolution of the system caused by action.
\item $S=\{1,2,\cdots,m\}$ is the register of the scheduler, which records the acting program.
\item $Q=(Q_1,Q_2,\cdots,Q_m)$ with $Q_i:\{0,1,\cdots,N\}\times L_i\mapsto 2^{\{l_{i,0},l_{i,1},\cdots,l_{i,n}\}}\times 2^{\{1,2,\cdots,m\}}\setminus\emptyset$ denoting the next location choice and the next program mapping.
\item $\rho_0$ is an initial state of the system, which lies in $\H$.
\item $l_0=(l_{1,i_{1,0}},l_{2,i_{2,0}},\cdots,l_{m,i_{m,0}})\in L$ with $l_{k,i_{k,0}}\in L_k$ being the initial location of Program $\pi_k$.
\item $s_0\in S$ denotes the acting program in the first round.
\end{itemize}
\end{defn}

Intuitively, this model admits $m$ programs running concurrently by $m$ processors. At each step of the concurrent system, one program is selected, and the statement at its location is executed. The statement contains a quantum measurement which helps to select the program to be executed in the next step.

Formally, the state of the system always lies in $\Delta(\H\times L_1\times L_2\times\cdots\times L_m\times [m])$ defined as
$$
\{\sum_{s=1}^m\sum_{i_1,\cdots,i_m} \rho_{s,i_1,\cdots,i_m}\otimes \op{l_{1,i_1}}{l_{1,i_1}}\otimes\cdots\otimes \op{l_{m,i_m}}{l_{m,i_m}}\otimes\op{s}{s}|\rho_{i_0,\cdots,i_m}\in\D(\H),\sum_{s,i_1,\cdots,i_m} \tr\rho_{s,i_0,\cdots,i_m}=1\}
$$
where $[m]=\{1,2,\cdots,m\}$.

The initial state is
$$
\sigma_0=\rho_0\otimes\op{l_{1,i_{1,0}}}{l_{1,i_{1,0}}}\otimes\cdots\otimes \op{l_{m,i_{m,0}}}{l_{m,i_{m,0}}}\otimes\op{s_0}{s_0}.
$$

At the first step, according to data $s_0$, $\pi_{s_0}$ is chosen and
 $\E_{s_0,l_{s_0,i_{s_0,0}}}$ is applied on the initial state to obtain
$$
\E_{s_0,l_{s_0,i_{s_0,0}}}(\rho_0)\otimes\op{l_{1,i_{1,0}}}{l_{1,i_{1,0}}}\otimes\cdots\otimes \op{l_{m,i_{m,0}}}{l_{m,i_{m,0}}}\otimes\op{s_0}{s_0}.
$$
Then measurement $\M_{l_{s_0,i_{s_0,0}}}$ is performed and based on the measurement outcome, classical index, and the corresponding location, the
location is changed according to $Q_{s_0}$. In other words, any state of the following form with $(f_{s_0,i_{s_0,0},j},t_{s_0,i_{s_0,0},j})\in Q_{s_0}(j,l_{s_0,i_{s_0,0}})$ is reachable non-deterministically
\begin{tiny}
$$
\sum_{j=0}^N M_{l_{s_0,l_{s_0,i_{s_0,0}}},j}\E_{s_0,l_{s_0,0}}(\rho_0)M_{l_{s_0,i_{s_0,0}},j}^{\dag}\otimes \op{l_{1,i_{1,0}}}{l_{1,i_{1,0}}}\otimes\cdots\otimes\op{f_{s_0,i_{s_0,0},j}}{f_{s_0,i_{s_0,0},j}}\otimes\cdots\otimes \op{l_{m,i_{m,0}}}{l_{m,i_{m,0}}}\otimes\op{t_{s_0,i_{s_0,0},j}}{t_{s_0,i_{s_0,0},j}},
$$
\end{tiny}
where only $\pi_{s_0}$'s location of all locations and the scheduler register would be changed.

Generally, at each step, the state $$\sum_{s=1}^m\sum_{s,i_1,\cdots,i_m} \rho_{s,i_1,\cdots,i_m}\otimes \op{l_{1,i_1}}{l_{1,i_1}}\otimes\cdots\otimes \op{l_{m,i_m}}{l_{m,i_m}}\otimes\op{s}{s}$$
is transformed by the following two sub-steps
\begin{itemize}
\item Apply super-operators according to the scheduler and location to obtain
$$\sum_{s=1}^m\sum_{i_1,\cdots,i_m} \E_{s,l_{s,i_s}}(\rho_{s,i_1,\cdots,i_m})\otimes \op{l_{1,i_1}}{l_{1,i_1}}\otimes\cdots\otimes \op{l_{m,i_m}}{l_{m,i_m}}\otimes\op{s}{s}.$$
\item Apply quantum measurement and change the location accordingly. The overall state becomes the following for any $(l_{s,i_s,j},t_{s,i_s,j})\in Q_s(j,l_{s,i_s})$ non-deterministically
    \begin{tiny}
$$
\sum_{j=0}^N\sum_{s=1}^m\sum_{i_1,\cdots,i_m} M_{l_{s,i_s},j}\E_{s,l_{s,i_s}}(\rho_{s,i_1,\cdots,i_m})M_{l_{s,i_s},j}^{\dag}\otimes  \op{l_{1,i_1}}{l_{1,i_1}}\otimes\cdots\otimes\op{l_{s,i_s,j}}{l_{s,i_s,j}}\otimes\cdots\otimes \op{l_{m,i_m}}{l_{m,i_m}}\otimes\op{t_{s,i_s,j}}{t_{s,i_s,j}}.
$$
\end{tiny}
\end{itemize}
\begin{defn}
We say $\omega=\sigma_0\sigma_2\cdots\sigma_k\cdots$ is admissible of the system $\pi$ if $\sigma_0=\rho_0\otimes\op{l_{1,i_{1,0}}}{l_{1,i_{1,0}}}\otimes\cdots\otimes \op{l_{m,i_{m,0}}}{l_{m,i_{m,0}}}\otimes\op{s_0}{s_0}$ is the initial state and $\sigma_i$ can be obtained by executing the system upon state $\sigma_{i-1}$ for all $i\geq 1$.
\end{defn}

This program becomes deterministic if and only if $|Q_s(j,l_{s,i_s})|=1$ for all $0\leq j\leq N$, $l_{s,i_s}\in L_s$ and $1\leq s\leq m$.

This model can also simulate a quantum program with exit by assuming an exit location $l_{s,e_s}\in L_s$ for each $s$ such that $\E_{l_{s,e_s}}$ and $\M_{l_{s,e_s}}$ do not change the state, and $Q_s(j,l_{s,e_s})=\{(l_{s,e_s},s)\}$ for each $0\leq j\leq N$ and $1\leq s\leq m$.

\section{QTL: Specifications and Their Classification}\label{sec:QTL}
To express the system properties and their development in time, we express the relations on states in a suitable language. When applied to programs, this will be a relation between the quantum data, the locations of all processors $\pi_1,\cdots,\pi_m$ together with the data of the scheduler.

The most general verification problem is to establish facts about the developments of the properties $q(\rho)$ in time by introducing time variables $t_1,t_2,\cdots\in \mathbb{N}$ as well as the time functional
\begin{align*}
H(t,q)\equiv q(\rho_t),
\end{align*}
where $\rho_t$ denotes the states, including the classical control components, in time $t_1,t_2,\cdots\in \mathbb{N}$.
Arbitrary time dependency can be expressed in the above formulism.

To illustrate our ideas without lengthy demonstration, we limit the expression power of the language with respect to dependency in time. More precisely, we only investigate basic predicates with single time variable and two time variables.
\subsection{Syntax of Quantum Temporal Logic}
As previously mentioned, we only use projections as atomic propositions $AP$ to build QTL, where $AP$ consists of all the operators of the following form
\begin{equation}
\sum_{s=1}^m\sum_{l_{1,i_1},l_{2,i_2},\cdots,l_{m,i_m}}P_{s,l_{1,i_1},l_{2,i_2},\cdots,l_{m,i_m}}\otimes\op{l_{1,i_1}}{l_{1,i_1}}\otimes\op{l_{2,i_2}}{l_{2,i_2}}\otimes\cdots\otimes\op{l_{m,i_m}}{l_{m,i_m}}\otimes\op{s}{s}
\end{equation}
where $P_{s,l_{1,i_1},l_{2,i_2},\cdots,l_{m,i_m}}$ are all projections of $\H$.
$AP$ contains two special elements, $I$ and $\{0\}$.

\begin{defn}
Let $p\in AP$ and $\rho\in\Delta(\H\times L_1\times L_2\times\cdots\times L_m\times [m])$. We say that $\rho$ satisfies $p$, written
$\rho\models p$, if $\supp(\rho)\subseteq p$; that is, $p\rho=\rho$.
\end{defn}
More precisely, QTL is built up from the logical operators $\wedge$ and $\vee$, the temporal modal $\mathbf{O}$ (next), $\mathbf{U}$ (until), $\tilde{\mathbf{U}}$ (almost surely until), $\mathbf{false}$, $\mathbf{true}$, $\lozenge$ (eventually), $\tilde{\lozenge}$ (almost surely eventually) and $\square$ (always).
The operators $\tilde{\mathbf{U}}$ and $\tilde{\lozenge}$ are introduced to characterize the asymptotical probabilistic behaviour induced by the quantum probability.

\begin{defn}
Formally, the set of QTL formulas over $AP$ is inductively defined as follows:
\begin{itemize}
\item    if $p\in AP$ then $p$ is an QTL formula;
\item  if $p\in AP$ then $\tilde{\lozenge}p$ is an QTL formula;
\item  if $p,q\in AP$ then $p\tilde{\mathbf{U}}q$ is an QTL formula;
\item    if $\phi$ and $\psi$ are QTL formulas then, $\phi\wedge \psi$, $\phi\vee \psi$, $\mathbf{O}\phi$, $\psi\mathbf{U}\phi$, $\lozenge \phi$, and $\square \phi$ are QTL formulas.
\end{itemize}
\end{defn}

Other than these fundamental operators, there are additional temporal operators defined in terms of the fundamental operators to write QTL formulas succinctly, for instance $\rightarrow$ and $\leftrightarrow$.

We do not allow $\neg$ because $\rho\models\neg p$ does not imply $\rho\models q$ for any $p,q\in AP$.
\subsection{Semantics of QTL}

A QTL formula can be satisfied by an infinite sequence of admissible states $w=\sigma_0\sigma_1\cdots\sigma_k\cdots$.
Let $w(i)=\sigma_i$, and $w^i=\sigma_i\sigma_{i+1}\cdots$
Formally, the satisfaction relation $\vDash$ between a sequence of states $\omega$ and an QTL formula is defined as follows:
\begin{itemize}
\item $w \vDash p$ {if} $w(0)\vDash p$;
\item $w\vDash\tilde{\lozenge}p$ for $p\in AP$ if for any $\delta>0$, there exists $i\geq 0$ such that $\tr[w(i)p]>1-\delta$;
\item $w \vDash q \tilde{\mathbf{U}} p$ for $p,q\in AP$ {if for any $\delta>0$ there exists} $i \geq 0$ such that $\tr[w(i)p]>1-\delta$ and for all $0 \leq k < i$, $w^k \vDash q$;
\item $w \vDash \phi \wedge \psi$ {if} $w \vDash \phi$ ~{and}~ $w \vDash \psi$\footnote{For $p,q\in AP$, $p\wedge q$ denotes the intersection of subspaces $p$ and $q$, $p\wedge q \in AP$.};
\item $w \vDash \phi \vee \psi$ {if} $w  \vDash \phi$ {or} $w \vDash \psi$ \footnote{For $p,q\in AP$, $p\vee q$ is the union of subspaces $p$ and $q$, $p\vee q$ is not always in $AP$. };
\item $w  \vDash \mathbf{O} \phi$ {if} $w^1 \vDash \phi$ ({in the next time step} $p$ {must be true});
\item $\lozenge \phi$ {if there exists} $i \geq 0$ such that $w^i \vDash \phi$;
\item $w \vDash \psi \mathbf{U} \phi$ {if there exists} $i \geq 0$ such that $w^i \vDash \phi$ and for all $0 \leq k < i$, $w^k \vDash \psi$ ($\psi$ must remain true until $\phi$ becomes true);
\item $w\square \phi$ if for any $i\geq 0$, $w^i\vDash \phi$.
\end{itemize}
We say an $\omega$-word $w$ satisfies a QTL formula $\phi$ when $w \vDash \phi$. The $\omega$-language $L(\phi)$ defined by $\phi$ is $\{w|w \vDash \phi,\forall~\mathrm{admissible}~w\}$, which is the set of $\omega$-admissible states that satisfy $\phi$. A formula $\phi$ is satisfiable if there exist $\omega$-admissible states $w$ such that $w \vDash \phi$.

The additional logical operators are defined as follows:
\begin{itemize}
\item $\phi\rightarrow \psi \equiv L(\phi)\subset L(\psi)$
\item $\phi\leftrightarrow \psi \equiv (\psi\rightarrow \phi)\wedge(\phi\rightarrow \phi)$
\item $\mathbf{true} \equiv I$,
\item $\mathbf{false} \equiv \{0\}$
\end{itemize}

\begin{defn}
For a quantum program $\pi$, we say that a QTL formula $\phi$ is valid if for any $\omega$-sequence of admissible state, $w$, we have $w\vDash \phi$.
\end{defn}

The reason that we use a sequence of quantum states rather than a sequence of subsets of $AP$ is to introduce the $\tilde{\lozenge}$ and $\tilde{\mathbf{U}}$ which study the asymptotical behaviour of the probability induced by quantum measurements. $\tilde{\lozenge} p$ describes the property that the induced number series  $a_0,a_1,\cdots,a_k,\cdots$ with $a_i=\tr[p\sigma_i]$ has, with $1$ as a limit point. The reason for introducing $\tilde{\mathbf{U}}$ is similar.

We reconsider Example \ref{example1} to illustrate the usefulness to introduce $\tilde{\lozenge} p$, as well as $\tilde{\mathbf{U}}$.
\begin{exam}
\begin{enumerate}
\item  $l_1$: $\rho_0=\op{-}{-}$, goto $l_2$;
 \item $l_2$: Measure $\rho$ using $\{M_0=\ket{0}\bra{0}, M_1=\ket{1}\bra{1}\}$, if outcome is $0$, goto $l_4$, otherwise, goto $l_3$;
 \item $l_3$: Apply $H$ gate on $\rho$, goto $l_2$
 \item $l_4$: goto $l_4;$
\end{enumerate}
\end{exam}
As illustrated in Example \ref{example1}, this program will not reach $l_4$ in finite steps. In other words, this program can not satisfy $\lozenge p$ with $p=\op{0}{0}\otimes\op{l_4}{l_4}$.
On the other hand, for any $\delta>0$, we can choose $n$ such that $\tr(\sigma_{2n}p)=1-\frac{1}{2^n}>1-\delta$. This program satisfies $\tilde{\lozenge} p$.

This example indicates that for quantum programs with exits
\begin{itemize}
\item $\lozenge$ and $\mathbf{U}$ are useful for tracking the total correctness of quantum programs which terminates in finite steps.
\item $\tilde{\lozenge}$ and $\tilde{\mathbf{U}}$ are useful for tracking the total correctness for quantum (probabilistic) programs which almost terminate.
\end{itemize}
In this example, we observe that $p$ is satisfies at any step with $p=\op{0}{0}\otimes\op{l_4}{l_4}+I_{\H}\otimes(\op{l_1}{l_1}+\op{l_2}{l_2}+\op{l_3}{l_3})$. In other words, $\square p$ is valid.
In general, $\square$ is useful for tracking the partial correctness of the program with exit. $\square p$ if we choose $p=\sum_{l_i\neq l_e\in L} I_{\H}\otimes\op{l_i}{l_i}+P\otimes \op{l_e}{l_e}$ where $P$ is the required property of output. It is invariantly true that whenever we
reach the exit, the output satisfies its specification.

\section{Example: Reasoning and Verification of Deterministic Quantum sequential Programs with Exits}\label{sec:example}

In this section, we focus on a special case of quantum concurrent programs--deterministic quantum programs with exits as an example.
We compare our model for deterministic quantum programs with exits with the widely studied Q-While language introduced in \cite{Ying11}. After reviewing the syntax and semantics of Q-While,
we show that our model for deterministic quantum programs with exits can be used for designing a compiler for Q-While.
Then, we prove a quantum B\"{o}hm-Jacopini theorem \cite{BJ66} which states that any deterministic quantum program with exit of our model is equivalent to a Q-While program. In particular, we only need to use a single Q-While statement on a larger space. Using this powerful tool, we are able to analyze such program very clearly.
\subsection{Q-While Language}
We first recall the syntax and semantics of Q-While.
\begin{defn}[Syntax \cite{Ying11}]\label{q-syntax}
The quantum \textbf{while}-programs are defined by
the grammar:
\begin{align*}
\label{syntax}S::=\ \mathbf{skip}\ & |\ S_1;S_2\ |\ q:=|0\rangle\ |\ \overline{q}:=U[\overline{q}] &|\ \mathbf{if}\ \left(\square m\cdot \M[\overline{q}] =m\rightarrow S_m\right)\ \mathbf{fi}\\ &|\ \mathbf{while}\ \M[\overline{q}]=1\ \mathbf{do}\ S\ \mathbf{od}
\end{align*}\end{defn}
$q:=|0\rangle$ means that quantum variable $q$ is initialised in a basis state $|0\rangle$. $\overline{q}:=U[\overline{q}]$ denotes that unitary transformation $U$ is applied to a sequence  $\overline{q}$ of quantum variables. In the case statement $\mathbf{if}\cdots\mathbf{fi}$, quantum measurement $\M$ is performed on $\overline{q}$ and then a subprogram $S_m$ is chosen for the next execution according to the measurement outcome $m$. In the loop $\mathbf{while}\cdots\mathbf{od}$, measurement $M$ in the guard has only two possible outcomes $0,1$: if the outcome is $0$ the loop terminates, and if the outcome is $1$, it executes the loop body $S$ and enters the loop again.

A configuration of a program is a pair $C=\langle S,\rho\rangle$
where $S$ is a program or the termination symbol $\downarrow$, and $\rho\in\mathcal{D}(\mathcal{H}_S)$ denotes the state of quantum system.

\begin{defn}[Operational Semantics \cite{Ying11}]\label{def-op-sem}
The operational semantics of quantum \textbf{while}-programs is defined as a transition relation $\rightarrow$ by the transition rules in the following.
\begin{equation*}\begin{split}&({\rm Sk})\ \ \langle\mathbf{skip},\rho\rangle\rightarrow\langle \downarrow,\rho\rangle \ \ ({\rm In})\ \ \ \langle
q:=|0\rangle,\rho\rangle\rightarrow\langle \downarrow,\rho^{q}_0\rangle \\
&({\rm UT})\ \ \langle\overline{q}:=U[\overline{q}],\rho\rangle\rightarrow\langle
\downarrow,U\rho U^{\dag}\rangle \ \ ({\rm SC})\ \ \ \frac{\langle S_1,\rho\rangle\rightarrow\langle
S_1^{\prime},\rho^{\prime}\rangle} {\langle
S_1;S_2,\rho\rangle\rightarrow\langle
S_1^{\prime};S_2,\rho^\prime\rangle}\\
&({\rm IF})\ \ \ \langle\mathbf{if}\ (\square m\cdot
M[\overline{q}]=m\rightarrow S_m)\ \mathbf{fi},\rho\rangle\rightarrow\langle
S_m,M_m\rho M_m^{\dag}\rangle\\
&({\rm L}0)\ \ \ \langle\mathbf{while}\
M[\overline{q}]=1\ \mathbf{do}\
S\ \mathbf{od},\rho\rangle\rightarrow\langle \downarrow, M_0\rho M_0^{\dag}\rangle \\
&({\rm L}1)\ \ \ \langle\mathbf{while}\
M[\overline{q}]=1\ \mathbf{do}\ S\ \mathbf{od},\rho\rangle\rightarrow\langle
S;\mathbf{while}\ M[\overline{q}]=1\ \mathbf{do}\ S\ \mathbf{od}, M_1\rho
M_1^{\dag}\rangle\end{split}\end{equation*}
\end{defn}
In (In), $\rho^{q}_0=\sum_n|0\rangle_q\langle n|\rho|n\rangle_q\langle
0|$. In (SC), we make the convention $\downarrow;S_2=S_2.$
In (IF), $m$ ranges over every possible outcome of measurement $M=\{M_m\}.$
Rules (In), (UT), (IF), (L0) and (L1) are determined by the basic postulates of quantum mechanics.

\begin{defn}[Denotational Semantics \cite{Ying11}]\label{den-sem-def} For any quantum \textbf{while}-program $S$, its semantic function is the mapping $\llbracket S\rrbracket:\mathcal{D}(\mathcal{H}_S)\rightarrow \mathcal{D}(\mathcal{H}_S)$ defined by \begin{equation}\llbracket S\rrbracket(\rho)=\sum\left\{\!|\rho^\prime: \langle S,\rho\rangle\rightarrow^\ast\langle \downarrow,\rho^\prime\rangle|\!\right\}\end{equation}
for every $\rho\in\mathcal{D}(\mathcal{H}_S)$, where $\rightarrow^\ast$ is the reflexive and transitive closure of $\rightarrow$, and $\left\{\!|\cdot|\!\right\}$ denotes a multi-set.
\end{defn}

\subsection{Deterministic Quantum Program with Exit}
The results given in this subsection are applicable for deterministic quantum program with exit, sequential or concurrent. To simplify the presentation, we only provide the proof detail for sequential quantum programs.
We first illustrate that our model can be used to design a compiler for Q-While.
\begin{theorem}\label{SW}
Any Q-While program can be expressed in the sequential quantum program model.
\end{theorem}
\begin{proof} The proof is given in Table \ref{Simulate}.
\begin{table}[h]
\begin{eqnarray*}\label{Simulate}
&&\mathbf{skip}\equiv\begin{cases}
l_1: \mathrm{goto}~ l_2;\\
 l_2: \cdots;
\end{cases}\ \ \ \ \ \ \ \ \ \ \ \ \ \ \ \ \
S_1;S_2\equiv\begin{cases}
 l_1: S_1, ~\mathrm{goto}~ l_2;\\
l_2: S_2, ~\mathrm{goto}~ l_3;\\
l_3: \cdots;
\end{cases}\\
&&q:=|0\rangle\equiv\begin{cases}
l_1: q:=|0\rangle, ~\mathrm{goto}~ l_2;\\
 l_2: \cdots;
\end{cases}\ \ \
q:=U[\overline{q}]\equiv
\begin{cases}
   l_1: \overline{q}:=U[\overline{q}], ~\mathrm{goto}~ l_2;\\
  l_2: \cdots;
\end{cases}\\
&&\mathbf{if}\ \left(\square m\cdot \M[\overline{q}] =m\rightarrow S_m\right)\ \mathbf{fi}\equiv\begin{cases}
   l_1: ~\mathrm{Measure} ~~\rho ~~\mathrm{in}~~ \M, ~\mathrm{if ~outcome~ is}~ m ~\mathrm{goto}~ l_{m+1};\\
  l_2: S_1, ~\mathrm{goto}~ l_{N+2};~~N~\mathrm{the~total~number~of~outcomes}.\\
  l_3: S_2, ~\mathrm{goto}~ l_{N+2};\\
  \cdots\\
  l_{N+1}: S_N, ~\mathrm{goto}~ l_{N+2};\\
  l_{N+2}:\cdots
\end{cases}\\
&&\mathbf{while}\ \M[\overline{q}]=1\ \mathbf{do}\ S\ \mathbf{od}\\
&\equiv&\begin{cases}
   l_1: ~\mathrm{Measure}~\rho ~~\mathrm{in}~~ \M, ~\mathrm{if ~outcome~ is}~ 0 ~\mathrm{goto}~ l_{3}, ~\mathrm{otherwise~goto}~ l_2;\\
  l_2: S, ~\mathrm{goto}~ l_{1};\\
  l_3: \cdots
\end{cases}\\
&&\mathbf{exit}:\equiv l_1: ~\mathrm{goto}~ l_1;
\end{eqnarray*}
\caption{Simulate Q-While}
\end{table}
\end{proof}

Suppose we have a deterministic program $\pi$ with locations $L$, where each location $l_{i}\in L$ is associated with a trace preserving operation $\E_{l_i}$ and a measurement $\M_{l_i}=\{M_{l_i,0},\cdots,M_{l_i,N}\}$, and a function $f:\{0,1,\cdots,N\}\times L\mapsto L$.

By considering the state space
$$
\Delta=\{\sum_{i=0}^n \rho_i\otimes \op{l_i}{l_i}:\rho_i\geq 0, \sum_{i=0}^n\tr(\rho_i)=1,\}
$$
Each step's operation $\E_{\pi}$ can be written as
$$
\E_{\pi}(\sum_{i=0}^n \rho_i\otimes \op{l_i}{l_i})=\sum_{j=0}^N\sum_{l_i\in L} M_{l_i,j}\E_{l_i}(\rho_i) M_{l_i,j}^{\dag}\otimes\op{f(j,l_i)}{f(j,l_i)}.
$$
In other words, $\E_{\pi}=\M\circ\E$ where
\begin{eqnarray*}
&&\E=\sum_{l_i}\E_{l_i}\otimes\op{l_i}{l_i},\\
&&\M(\cdot)=\sum_{j=0}^N\sum_{l_i\in L}(M_{l_i,j} \otimes\op{f(j,l_i)}{l_i})\cdot (M_{l_i,j}^{\dag}\otimes\op{l_i}{f(j,l_i)}).
\end{eqnarray*}
Similarly, a deterministic quantum concurrent program can be modeled by a super-operator.
Therefore, we have the following lemma.
\begin{lem}\label{l1}
 A deterministic quantum program $\pi$ can be modeled by a super-operator $\E_{\pi}$ in a larger state space $\Delta$. After $k$ step, the state becomes $\E_{\pi}^k(\sigma_0)$ where $\sigma_0=\rho_0\otimes\op{l_0}{l_0}$.
\end{lem}
 For programs with exits, we have the following:
\begin{lem}\label{l2}
For deterministic quantum program, which contain an exit location $l_e$, we have the following inequality
\begin{equation}
(I_{\H}\otimes \op{l_e}{l_e})\E(\sigma)(I_{\H}\otimes\op{l_e}{l_e})\geq (I_{\H}\otimes \op{l_e}{l_e})\sigma(I_{\H}\otimes\op{l_e}{l_e})
\end{equation}
which holds for any $\sigma\in\Delta$.
\end{lem}
\begin{proof}
According to the definition of exit location, we have
$$
\E_{\pi}(\rho_e\otimes\op{l_e}{l_e})=\rho_e\otimes\op{l_e}{l_e}.
$$
We note the following
\begin{eqnarray*}
&&(I_{\H}\otimes \op{l_e}{l_e})\E(\sigma)(I_{\H}\otimes\op{l_e}{l_e})- (I_{\H}\otimes \op{l_e}{l_e})\sigma(I_{\H}\otimes\op{l_e}{l_e})\\
&=& (I_{\H}\otimes \op{l_e}{l_e})\E[\sigma-(I_{\H}\otimes \op{l_e}{l_e})\sigma(I_{\H}\otimes\op{l_e}{l_e})](I_{\H}\otimes\op{l_e}{l_e})\\
&\geq &0,
\end{eqnarray*}
where we use the fact that any state $\sigma\in\Delta$, $\sigma-[I_{\H}\otimes \op{l_e}{l_e})\sigma(I_{\H}\otimes\op{l_e}{l_e})\geq 0$.
\end{proof}
The above statement is also true for quantum concurrent program $\pi$ consisting of $\pi_1,\pi_2,\cdots,\pi_m$ where each $\pi_s$ has an exit location $l_{s,e_s}$. where the only difference is that there is more than one exit location.

According to Lemma \ref{l1} and Lemma \ref{l2}, we have the following quantum B\"{o}hm-Jacopini theorem.
\begin{theorem}\label{BJ}
Any deterministic program $\pi$ with exit, including program written in Q-While, can be written as a single ``while'' statement in Q-While for input $\sigma_0\vDash I_{\H}\otimes\op{l_0}{l_0}$,
\begin{equation}
\mathbf{while}\ \M[\overline{q}]=1\ \mathbf{do}\ S\ \mathbf{od},
\end{equation}
where $S=\E_{\pi}$, $\M=\{M_0=I_{\H}\otimes\op{l_e}{l_e},M_1=I_{\H}\otimes\sum_{l_i\neq l_e,l_i\in L}\op{l_i}{l_i}\}$.
\end{theorem}
\subsubsection{Proof System for Partial and Total Correctness}
We borrow the idea from \cite{ZYY19} to derive a proof system reasoning $\tilde{\lozenge}$ and $\square$ for program $\pi$ with exit.
We first generalize the definition of partial correctness and total correctness of Q-While program \cite{ZYY19} to deterministic program $S$ with exit.
\begin{defn}
\begin{enumerate}
\item $\{P\}S\{Q\}$ is true in
the sense of partial correctness in aQHL, written:
$\models_{\rm par}^{\rm a}\{P\}S\{Q\}$, if for all
$\rho\vDash P$: $$ \lb S\rb (\rho)\models Q.$$
\item $\{P\}S\{Q\}$ is true in
the sense of total correctness in aQHL, written:
$\models_{\rm tot}^{\rm a}\{P\}S\{Q\}$, if for all $\rho\vDash P$: $$\lb S\rb (\rho)\models Q\ \&\ \tr(\lb S\rb (\rho))=\tr \rho.$$
\end{enumerate}
\end{defn}
\begin{fact}
For any deterministic program $S$ with exit location $l_e$, we employ Theorem \ref{BJ} to transform it into a Q-While program $\pi$ in state space $\H\otimes L$. We have the following correspondence
\begin{itemize}
\item $\models_{\rm par}^{\rm a}\{P\}S\{Q\}$ is equivalent to $\square q$ in $\pi$ by choosing $q=\sum_{l_i\neq l_e\in L} I_{\H}\otimes\op{l_i}{l_i}+Q\otimes \op{l_e}{l_e}$ for input $\sigma_0\vDash P\otimes\op{l_0}{l_0}$.
\item $\models_{\rm tot}^{\rm a}\{P\}S\{Q\}$ is equivalent to $\tilde{\lozenge} q$ in $\pi$ by choosing $q=Q\otimes \op{l_e}{l_e}$ for input $\sigma_0\vDash P\otimes\op{l_0}{l_0}$.
\end{itemize}
\end{fact}
Due to this correspondence, we can derive a relatively complete proof system for general deterministic program with exit using the results of \cite{ZYY19}.
\begin{prop}
For $\pi=\mathbf{while}\ \M[\overline{q}]=1\ \mathbf{do}\ S\ \mathbf{od}$ with $S=\E_{\pi}$, $\M=\{M_0=I_{\H}\otimes\op{l_e}{l_e},M_1=I_{\H}\otimes\sum_{l_i\neq l_e,l_i\in L}\op{l_i}{l_i}\}$ and $p,q\in AP$,
the following proof system is relatively complete for proving $\square q$ and $\tilde{\lozenge} q$ respectively,
\begin{equation*}\begin{split}
&\square:~~~~
\frac{\{p\}S\{\supp[(M_0\wedge q)+ (M_1\wedge p)]\}}{\{\supp[(M_0\wedge q)+ (M_1\wedge p)]\}\pi\{q\}} \\
&\tilde{\lozenge}:~~~~
\frac{\begin{split}\{&p\}S\{\supp[(M_0\wedge q)+ (M_1\wedge p)]\}\\
{\rm for\ any}\ \epsilon>0,\ t_\epsilon\ {\rm is\ a}\ &
(\supp[(M_0\wedge q)+ (M_1\wedge p)],\epsilon){\rm {\text -}ranking}~{\rm function\ of\ }
\mathbf{while}\end{split}}{\{\supp[(M_0\wedge q)+ (M_1\wedge p)]\}\pi\{q\}}
\end{split}\end{equation*}
where a function $t: \mathcal{D}(\mathcal{H}_\mathbf{while})\rightarrow\mathbb{N}$ is called a $(q,\epsilon)$-ranking function of $\mathbf{while}$ if for all
$\sigma$ with $\sigma\models q$, we have $\lb S\rb (M_1\rho{M_1}) \models q$,
 $t(\lb S\rb (M_1\rho{M_1}))\leq t(\sigma)$ and  $\tr(\sigma)\geq\epsilon$ implies
$t(\lb S\rb (M_1\sigma{M_1}))< t(\sigma)$.
\end{prop}
\subsubsection{Compute the reachability super-operator for finite dimensional systems}
By Theorem \ref{BJ}, we have
\begin{theorem}\label{compute}
Suppose the operation of $\pi$ is $\E_{\pi}$ with rational entries on quantum system $\H$ and locations $L$ with $l_e$ as the exit location.
The reachability super-operator is defined as
$$\F_{\pi}(\cdot)=\lim_{n\rightarrow\infty} (I_{\H}\otimes \op{l_e}{l_e})\E_{\pi}^n(\cdot)(I_{\H}\otimes \op{l_e}{l_e})$$
and can be computed in polynomial time.

Moreover, the average running time of the program can be computed in polynomial time for any rational input $\sigma_0=\rho_0\otimes\op{l_0}{l_0}$.
\end{theorem}
\begin{proof}
According to Theorem \ref{BJ}, we reformulate the program in
\begin{equation*}
\mathbf{while}\ \M[\overline{q}]=1\ \mathbf{do}\ S\ \mathbf{od},
\end{equation*}
where $S=\E_{\pi}$, $\M=\{M_0=I_{\H}\otimes\op{l_e}{l_e},M_1=I_{\H}\otimes\sum_{l_i\neq l_e,l_i\in L}\op{l_i}{l_i}\}$.

From \cite{YYFD13}, we know that $\F_{\pi}(\cdot)=\sum_{i} F_i\cdot F_i^{\dag}$
satisfies
\begin{equation}
\sum_{i} F_i\otimes F_i^{*}=(M_0\otimes M_0)(I_{\H\otimes L\otimes\H\otimes L}-N)^{-1},
\end{equation}
where $N$ is obtained by replacing $M$'s Jordan blocks with eigenvalues $\lambda$ for $|\lambda|=1$ with $M=(\sum_{i} E_i\otimes E_i^{*})(M_1\otimes M_1)$.

Moreover, the average running time of the while program equals the average running time of the original program for reaching the exit location, which can be written as
\begin{equation}
\langle\Phi|(M_0\otimes M_0)(I-N)^{-2}(\sigma_0\otimes I_{\H\otimes L})|\Phi\rangle.
\end{equation}
where $|\Phi\rangle=\sum_{i,l}\ket{il}\ket{il}$ is the unnormalized maximally entangled state on $\H\otimes L\otimes\H\otimes L$.

Note that the Jordan block can be computed in the polynomial time of the input size \cite{CAI94}, also verifying whether $|\lambda|=1$ is valid in polynomial time for algebraic number $\lambda$ \cite{GE93}.
\end{proof}
Remark: Before this work, this is not known, even for the Q-While programs, because of the complexity of the composition of algebraic numbers due to the composition of whiles.

For deterministic program with exit, $\square$, $\tilde{\lozenge}$ and $\lozenge$ can be verified.
\begin{theorem}\label{test}
Suppose the operation of $\pi$ is $\E_{\pi}$ with rational entries on $d$-dimensional quantum system $\H$ and locations $L$ with $l_e$ as exit location, $|L|=l$ .
For $p=P\otimes\op{l_e}{l_e}\in AP$, $\tilde{\lozenge} p$, $\square p$, ${\lozenge}p$ can be verified in polynomial time.
\end{theorem}
\begin{proof}
$\tilde{\lozenge} p$ can be verified using Theorem \ref{compute}.

${\lozenge}p$ if the program terminates in finite steps and satisfies $p$. According to \cite{YYFD13}, a While program on $dl$ dimensional system terminates if and only if it terminates in $dl-1$ steps.
The rest is to verify $p$ on $\E_{\pi}^{d-1}(\rho_0\otimes\op{l_0}{l_0})$.

$\square p$ if and only if $\sum_{k=0}^{dl-1} {\E_{\pi}}^k(\sigma_0)/dl \vDash p$. The only if part is due to the definition of $\square p$:
for any $k$ and any $\sigma_0=\rho_0\otimes\op{l_0}{l_0}$, ${\E_{\pi}}^k(\sigma_0)\vDash p$. Then, $\sum_{k=0}^{dl-1} {\E_{\pi}}^k(\sigma_0)/dl \vDash p$. To see the converse,
we let $P_k=\supp(\sum_{k=0}^{k} {\E_{\pi}}^k(\sigma_0)/dl)  \vDash p$, then $P_0\subseteq P_1\subseteq \cdots \subseteq P_{d-1}\subseteq \H$. Moreover, if $P_{k}=P_{k+1}$, then $P_{k+1}=P_{k+2}$. By counting the dimensions, we know that the increasing chain of subspaces converge after $dl-1$ steps.
\end{proof}

\subsubsection{Quantum Kleene Closure}
Note that the quantum state can always be entangled with other systems. For this situation, we can have the following result as the quantum analogue of the Kleene closure.
\begin{theorem}
Suppose the operation of $\pi$ is $\E_A$ with rational entries on $d$-dimensional system $\H_A$, and the input state is given as $\rho_{AB}$ which is a general state in system $\H_A\otimes\H_B$. $\E_A\otimes \mathcal{I}_B$ is an operator applied on $\H_A\otimes\H_B$ with input $\rho_{AB}$.
Suppose $p$ is a projection of $\H_A\otimes\H_B$, we have
$$
\square p \Leftrightarrow \sum_{k=0}^{t-1}\frac{(\E_A^k\otimes \mathcal{I}_B)(\rho_{AB})}{t} \vDash p
$$
where $t$ can be chosen to be $d^2-1$, and does not need to depend on the dimension of $\H_B$.
\end{theorem}
\begin{proof}
By the fact of $\supp(\rho/2+\sigma/2)=\mathrm{span}\{\supp(\rho),\supp(\sigma)\}$, we only need to consider $\rho_{AB}=\op{\chi}{\chi}$ to be a pure state. Let $\ket{\chi}$ have a Schmidt decomposition $\ket{\chi}=\sum_{k=1}^s \nu_i\ket{\alpha_i}\ket{\beta_i}$ \cite{NI11} with $s\leq d$, then $(\E_A^k\otimes \mathcal{I}_B)(\rho_{AB})\vDash q=\H_A\otimes Q$ with $Q$ being the subspace of $\H_B$ spanned by $\ket{\beta_i}$. This problem is transformed into system with dimension no more than $d^2$. The rest is the same as the proof of Theorem \ref{test}.
\end{proof}
This problem of determining the smallest $t$ is related to the subalgebra generation problem \cite{LA86}, which has recently been shown to
be less than $\sqrt{2} d^{1.5}+3d$ \cite{GUT18,PAP97}. This implies that $t$ of the above theorem can be chosen to be $\sqrt{2} d^{1.5}+3d$. Very recently, this bound is improved into $2d\log_2 d +4d$ by Shitov in \cite{Shitov19}.

It is interesting that $t$ can be chosen independent of the dimension of $\H_B$, even $\H_B$ is infinite dimensional.

\section{Decidability of QTL}\label{sec:complexity}
According to the analysis of the last three sections, each deterministic quantum program is determined by a trace preserving super-operator. For non-deterministic quantum program (sequential or concurrent), each non-deterministic choice is correlated to a super-operator. Therefore, a non-deterministic quantum program can always be modeled as a quantum automaton.
\begin{defn}
A quantum automaton $\mathcal{A}=(\H,Act,\{\E_{\alpha}|\alpha\in Act\},\rho_0)$, where
\begin{itemize}
\item $\H$ is a Hilbert space, called the state space;
\item $Act$ is a set of finite action names;
\item for each $\alpha\in Act$, $\E_{\alpha}$ is a trace preserving super-operator;
\item $\rho_0$ is an initial state of the system, which lies in $\H$.
\end{itemize}
At the first step, $\E_{\alpha_0}$ is applied on $\rho_0$. After this, $\E_{\alpha}$ is chosen non-deterministically at each step.
\end{defn}
For simplicity, we choose $\sigma_0=\E_{\alpha_0}(\rho_0)$ as an initial state, and each $w=w_1w_2\cdots w_k\cdots\in Act^{\omega}$ induces an admissible state sequence $\sigma_0\sigma_1\sigma_2\cdots\sigma_k\cdots$ with $\sigma_i=\E_{w_i}(\sigma_{i-1})$ for all $i\geq 1$.

To study the QTL of a general quantum concurrent system, we only need to study the QTL of the corresponding quantum automaton.

The following observation from \cite{YY12} is useful.
\begin{fact}\label{inverse}
For any $p\in AP$ and super-operator $\E$, we have
\begin{equation}
\{\sigma|\E(\sigma)\vDash p\}=\{\sigma|\sigma\vDash (\E^*(p^{\perp}))^{\perp}\},
\end{equation}
where for any non-negative matrix $M$, $M^{\perp}=\{\supp(M)\}^{\perp}$ where $\supp(M)$ denotes the subspace spanned by the eigenvectors of $M$ with nonzero eigenvalues.

For $p\in AP$, let $\E^{-1}(p):=(\E^*(p^{\perp}))^{\perp}$, we have $\E^{-1}(\vee_{i=1}^r p_i)=\vee_{i=1}^r \E^{-1}(p_i)$.
For $p\in AP$, let $\E(p)=\supp(\E(\frac{p}{d_p}))$, where $d_p$ denotes the dimension of $p$, we have $\E(\vee_{i=1}^r p_i)=\vee_{i=1}^r \E(p_i)$.
\end{fact}
In the following, we study the decidability of the basic QTL formulae.
We assume all the entries are rational numbers, and $\H$ is $d$-dimensional.
If all $\E_{\alpha}$ are all unitaries, the decidability of $\square \vee_{i=1}^n p_i$, $\square\lozenge \vee_{i=1}^n p_i$ and $\lozenge \square\vee_{i=1}^n p_i$ was presented in \cite{LY14}. As part of our results below, we generalizing these results into general super-operators  and thus solve the open questions of \cite{LY14}.

It worth to notice that $\square\lozenge \vee_{i=1}^n p_i$ and $\lozenge \square\vee_{i=1}^n p_i$ are highly nontrivial in the sense that no uniform time bound for
the $\lozenge$ on different paths.

\subsection{Next}
To decide $\mathbf{O}(p)$, we only need to verify $\vDash p$ for the next step. In other words, it is enough to verify $\vDash p$ for the original system with different initial state $\sigma=\E_{\alpha}(\sigma_0)$ for all $\alpha\in Act$ where $\sigma_0$ is the current state.
\subsection{Invariance}
One-time variable invariance is universally quantified.
Invariance, denoted by $\square$, is a property holding throughout all states of all possible execution sequences.
\begin{lem}\cite{YY12}.
For $p\in AP$ and $|Act|>1$, we know that $\square p$ if and only if $\E^k(\sigma_0)\vDash p$ for $0\leq k\leq d-1$ with $\E=\frac{\sum_{\alpha\in Act} \E_{\alpha}}{|Act|}$.
\end{lem}
If $\phi$ is not an element of $AP$, but a union of elements in $AP$, in other words, $\phi$ is a finite union of closed subspaces, the problem of determining $\square \phi$ becomes non-trivial. In the following, we consider the general case.
\begin{theorem}\label{inv}
If $\phi=\vee_{i=1}^n p_i$ with $p_i\in AP$, $\square \phi$ is decidable for $|Act|>1$.
\end{theorem}
\begin{proof}
We can characterize all initial states $\sigma_0$ such that $\square \phi$ is valid in the following.
In particular, we present $\psi=\vee_{j=1}^r q_i$ with $q_i\in AP$ and show that $\square \phi$ if and only if $\sigma_0\vDash \psi$.
To see this, we let
\begin{eqnarray*}
Y_0&=& \phi,\\
Y_1&=&Y_0\bigcap_{\alpha\in Act} \E_{\alpha}^{-1}(Y_0),\\
&\cdots&,\\
Y_k&=&Y_{k-1}\bigcap_{\alpha\in Act} \E_{\alpha}^{-1}(Y_{k-1}),\\
&\cdots&
\end{eqnarray*}
$Y_0$ characterizes the set of states that satisfies $\phi$ in $0$ step. $Y_1$ characterizes the set of states that satisfies $\phi$ in $0$ step and $1$ step. $Y_k$ characterizes the set of states that satisfies $\phi$ in less than $k+1$ step.
Each $Y_i$ is a finite union of subspaces, and $Y_0\supseteq Y_1\supseteq \cdots Y_k\supseteq$. According to Lemma \ref{uss}, there exists $k$ such that $Y_k=Y_{m}$ for all $m\geq k$.
Let $Y_k=\psi$, then, we only need to verify $\sigma_0\vDash \psi$.
\end{proof}
\subsection{Eventually and Almost Surely Eventually}
Eventually is another one time variable. Determining $\lozenge p$ is an open issue, even for $|Act|=1$ and $p\in AP$. This is related to the famous Skolem problem \cite{Tao07} which asks whether there exists $n$ such that $A^n x\in p$ for given rational square matrix $A$, vector $x$ and subspace $p$.
\begin{fact}
For $|Act|=1$, $\lozenge p$ is decidable for $p\in AP$ if and only if the Skolem Problem \ref{Skolem} is decidable.
\end{fact}
\begin{proof}
For any $A$ and $x$ without loss of generality, assume $A^{\dag}A\leq I_{H}$ and $x^{\dag}x\leq 1$, and $\ket{0}\in p$. We design the following trace preserving super-operator
$$
\E(\rho)=A\rho A^{\dag}+(1-\tr A^{\dag}A\rho)\op{0}{0}.
$$
$\E^n(xx^{\dag})\vDash p$ if and only if $A^n x\in p$.

On the other hand, if we can solve the Skolem problem, we can verify $\lozenge p$ for $|Act|=1$ as follows.
Let $q=I-p$ where $p\in AP$ is regarded as the projection onto subspace $p$.
\begin{eqnarray*}
\E^n(\rho_0)\vDash p &\Leftrightarrow& \tr[\E^n(\sigma_0)p^{\perp}]=0 \\
&\Leftrightarrow& \langle \Phi|(q\otimes I)M^n(\sigma_0\otimes I)|\Phi\rangle=0\\
&\Leftrightarrow& M^n(\sigma_0\otimes I)|\Phi\rangle \in \{\ket{v}| \ket{v}\perp (q\otimes I)\ket{\Phi}\},
\end{eqnarray*}
where $M$ is the matrix representation of $\E$ and $\ket{\Phi}=\sum_{i=1}^d\ket{ii}$.
\end{proof}
Concerning $\tilde{\lozenge} p$, we have a similar result as $\lozenge p$.
\begin{fact}
For $|Act|=1$, $\tilde{\lozenge} p$ is decidable for $p\in AP$ if and only if the Skolem problem \ref{Skolem} is decidable.
\end{fact}
\begin{proof}
$\tilde{\lozenge} p$ if and only if one of the following two cases is valid.
\begin{itemize}
\item $\lozenge p$;
\item There exists a sequence of $n_1,n_2,\cdots,n_k$ such that the induced number sequence $a_1,a_2,\cdots,a_k,\cdots$ converges to $1$ where $a_k=\tr[p\E^{n_k}(\sigma_0)]$.
\end{itemize}
Case 2 is decidable via the following procedure.
Let the linear recurrent series $b_n=1-a_n=\tr[(I-p)\E^{n}(\sigma_0)]$.
According to Subsection \ref{MRSO}, we first write it as $b_n=\tr(M^n N)$ where $M=\sum_j E_j\otimes E_j^{*}$ and $N$ is determined by $(I-p)$ and $\sigma_0$.
Let $M=S J(M) S^{-1}$ be the Jordan decomposition of $M$, we have $b_n=\tr[J(M)^nS^{-1}NS]$.
According to Lemma \ref{tech0}, we know that every Jordan block of $M$ has eigenvalue with absolution no more than 1, and the size of the block with absolution 1 eigenvalue is 1.
As we are only interested in the limit points, we only need to care about number series $c_n=\tr[J(M)'^nS^{-1}NS]$, where we delete all the Jordan blocks whose absolution of eigenvalue is smaller than 1. In other words, $J'=J(M)'$ is a diagonal matrix with eigenvalues either 0 or absolution 1.
The problem becomes seeing whether $0$ is a limit point of $c_n$.

Note that these eigenvalues are all algebraic numbers. Without loss of generality, we assume these nonzero eigenvalues $\lambda_1,\cdots,\lambda_r$ lie at the leading principal submatrix.
According to Theorem \ref{K2}, one can compute a basis of lattice
$$
La:=\{(k_1,\cdots,k_r)|\Pi_{j=1}^r \lambda_j^{k_j}=1\}\subset\mathbb{Z}^r,
$$
in polynomial time.
Assume the computed basis is $v_1,v_2,\cdots,v_t$ with $v_j=(v_{j,1},\cdots,v_{j,r})^T\in \mathbb{Z}^r$.
According to Theorem \ref{K1}, the closure of $\{J'^0,J'^1,\cdots\}$, the diagonal elements, is characterized by
$$
\{(e^{i\theta_1},e^{i\theta_2},\cdots,e^{i\theta_{r}})|\Pi_{k=1}^r e^{i\theta_k v_{j,k}}=1~\forall~1\leq j\leq t.\}
$$
Now the problem becomes to determine whether there exists $(e^{i\theta_1},e^{i\theta_2},\cdots,e^{i\theta_{r}})$ lying in the above set which satisfies
$\tr[JS^{-1}NS]=0$ with $J=diag\{e^{i\theta_1},e^{i\theta_2},\cdots,e^{i\theta_{r}},0,\cdots,0\}$.
Let $e^{i\theta_k}=x_k+y_k$, then $\tr[JS^{-1}NS]=0$ is a polynomial equation of $x_k$ and $y_k$ with algebraic coefficients. $\Pi_{k=1}^r e^{i\theta_k v_{j,k}}=1$ can also be rewritten as polynomial equations of $x_k$ and $y_k$ with algebraic coefficients. The rest of the equations are $x_k^2+y_k^2=1$ for $1\leq k\leq t$.
According to Tarski's theorem on polynomial equations with algebraic coefficients, it is decidable to verify whether $0$ is a limit point of $c_n$.
\end{proof}
Note that Case 2 of the above proof actually shows
\begin{fact} \label{tilde1}
For $|Act|=1$ and $p\in AP$, $\square \tilde{\lozenge} p$ is decidable.
\end{fact}
We can further show that
\begin{fact}\label{tilde2}
For $|Act|=1$ and $p\in AP$, $\square{\lozenge} p$ is decidable.
\end{fact}
\begin{proof}
$\square{\lozenge} p$ if and only if there exist infinite $n$ such that linear recurrent $b_n=\tr[(I-p)\E^{n}(\rho_0)]=0$. We can always find an integer $s$ such that
$c_n=s^nb_n=\tr[A^nM]$ for integer matrices $A$ and rank 1 matrix $M$. $b_n=0$ if and only if $c_n=0$.
If this is true, by Theorem \ref{SML}, we can bound the period $r\leq p^{d^4}$ by choosing $p\nmid 2\det{A}$.
The rest is to verify whether $c_{vr+u}=0$ is valid for fixed $u\leq r$ and all $v$. This can be done as for fixed $r$ and $u$, $d_v=c_{vr+u}$ is still a linear recurrent series of degree $d^2$. $d_v\equiv 0$ if and only if the initial $d^2+1$ element is $0$.
\end{proof}

To show the decidability of $\lozenge\square \phi$ with $\phi=\vee_{i=1}^t p_i$ and $p_i\in AP$ for $|Act|>1$, the following lemmas are needed.

\begin{lem}\label{invariant}
For finite union of subspaces $r\subseteq\H$, we can construct $x\subseteq r$ in finite steps, such that
\begin{itemize}
\item $\vee_{\alpha} \E_{\alpha}(x)=x$;
\item For any $x'\subseteq r$ such that $\vee_{\alpha} \E_{\alpha}(x')=x'$, we always have $x'\subseteq x$.
\end{itemize}
Moreover, such $x$ is also a finite union of subspaces, and $x$ is called the maximal invariant of $r$.
\end{lem}
Remark: All $x\subseteq \H$ can always be written as a union (a possibly infinite union or even a continuous union), for this sense, $\vee$ is still well defined as a union.
\begin{proof}
Our construction of $x$ is as follows:
\begin{eqnarray*}
&&Z_0=r,\\
&&Z_1=Z_0\bigcap_{\alpha\in Act} \E_{\alpha}^{-1}(Z_0)\bigcap[\vee_{\alpha\in Act}\E_{\alpha}(Z_0)],\\
&&\cdots,\\
&&Z_k=Z_{k-1}\bigcap_{\alpha\in Act} \E_{\alpha}^{-1}(Z_{k-1})\bigcap[\vee_{\alpha\in Act}\E_{\alpha}(Z_{k-1})],\\
&&\cdots
\end{eqnarray*}
First note that each $Z_i$ is a finite union of subspaces and forms a decreasing chain
$$p=Z_0\supseteq Z_1\supseteq \cdots\supseteq Z_k\supseteq\cdots$$
According to Lemma \ref{uss}, there exists $n$ such that $Z_n=Z_m$ for any $m\geq n$.
Let $x=Z_n$, we have
\begin{eqnarray*}
&&x=x\bigcap_{\alpha\in Act} \E_{\alpha}^{-1}(x)\bigcap[\vee_{\alpha\in Act}\E_{\alpha}(x)]\\
&\Rightarrow & x\subseteq\E_{\alpha}^{-1}(x),~x\subseteq \vee_{\alpha\in Act}\E_{\alpha}(x)\\
&\Rightarrow & \E_{\alpha}(x)\subseteq x,~~x\subseteq \vee_{\alpha\in Act}\E_{\alpha}(x)\\
&\Rightarrow & \vee_{\alpha\in Act}\E_{\alpha}(x)\subseteq x,~x\subseteq \vee_{\alpha\in Act}\E_{\alpha}(x)\\
&\Rightarrow & x=\vee_{\alpha\in Act}\E_{\alpha}(x).
\end{eqnarray*}
Moreover, $x$ is a finite union of subspaces.

Assume $\vee_{\alpha\in Act}\E_{\alpha}(x')= x'\subseteq r=Z_0$, we have
\begin{eqnarray*}
&&\E_{\alpha}(x')\subset Z_0,~x'\subseteq Z_0,~\\
&\Rightarrow& x'\subseteq \E_{\alpha}^{-1}(Z_0), \E_{\alpha}(x')\subseteq \E_{\alpha}(Z_0)\\
&\Rightarrow& x'\subseteq \bigcap_{\alpha\in Act}\E_{\alpha}^{-1}(Z_0),~x'=\vee_{\alpha\in Act}\E_{\alpha}(x')\subseteq \vee_{\alpha\in Act}\E_{\alpha}(Z_0)\\
&\Rightarrow & x'\in Z_1 ~~\Rightarrow~~\cdots~~\Rightarrow~~x'\subseteq Z_k~~\Rightarrow~ x'\subseteq x.
\end{eqnarray*}
This completes the proof.
\end{proof}

\begin{lem}\label{extension}
For finite union of subspaces $x\subseteq \H$ with $\vee_{\alpha\in Act}\E_{\alpha}(x)=x$, we can construct $y$ in finite steps such that
\begin{itemize}
\item $x\subseteq y$ and $y=\bigcap_{\alpha\in Act} \E_{\alpha}^{-1}(y)$.
\item For any $y'\subseteq \H$ such that $y'=\bigcap_{\alpha\in Act} \E_{\alpha}^{-1}(y')$ and $x\subseteq y'$ we always have $y\subseteq y'$.
\end{itemize}
Moreover, such $y$ is also a finite union of subspaces. We call $y$ the maximal extension of $x$.
\end{lem}
\begin{proof} Observe that $y$ is characterized as the limit by the following,
\begin{eqnarray*}
Y_0&=&x,\\
Y_1&=&\bigcap_{\alpha\in Act} \E_{\alpha}^{-1}(Y_0),\\
\cdots\\
Y_k&=&\bigcap_{\alpha\in Act} \E_{\alpha}^{-1}(Y_{k-1}),\\
\cdots
\end{eqnarray*}
Since $\E_{\alpha}(x)\subseteq x=Y_0$, then $x\subseteq \E_{\alpha}^{-1}(x)$. Thus, $x=Y_0\subseteq \bigcap_{\alpha\in Act} \E_{\alpha}^{-1}(Y_0)=Y_1$. By repeating this argument, we know that $Y_0\subseteq Y_1\subseteq \cdots\subseteq Y_k\subseteq$. $Y_k$ characterize the set of input states such that $x$ is satisfied in the $k$-th steps.
Unfortunately, with the general increasing of the chain of finite union of subspaces, there is no guarantee of termination.

To show the termination of the series, we first note that for each $\alpha$, and $z=\vee_{i=1}^s z_i$ with $z_i\in AP$, $\E_{\alpha}^{-1}(z)$ is still a union of no more than $s$ subspaces. In other words, the number of unions is not increasing in computing the pre-image. Now we can construct an $|Act|$-ary tree $T$ whose nodes are all the union of at most $s$ subspaces.
\begin{itemize}
\item Let $x=\vee_{i=1}^t x_i$ be the root of $T$ where $x_i\in AP$.
\item At the first step, for each $\alpha$, we generate $\E_{\alpha}^{-1}(x)$. If $\E_{\alpha}^{-1}(x)\supsetneq x$, we add $\E_{\alpha}^{-1}(x)$ as a child of $x$. Otherwise $\E_{\alpha}^{-1}(x)=x$, we mark $x$ as a "star" node. Now we have a tree of height at most $2$.
 \item   $\cdots$
 \item At the $k$-th step, for each $\alpha$ and each leaf node $n_d$ of the current tree, we generate $\E_{\alpha}^{-1}(n_d)$. If $\E_{\alpha}^{-1}(n_d)\supsetneq n_d$, we add $\E_{\alpha}^{-1}(n_d)$ as a child of $n_d$. Otherwise $\E_{\alpha}^{-1}(n_d)=n_d$, we mark $n_d$ as a "star" node.
 \item $\cdots$
 \item Stop if current leaves are all "star" nodes.
\end{itemize}
One can verify that this tree is a strictly increasing tree, in the sense that each node (union of at most $s$ elements) strictly contains its parent node.
We can easily verify that the height of this tree is at most $td$ according to the following fact via simply counting the dimension.
For each strictly increasing chain of $Z_0\subsetneq Z_1\subsetneq \cdots\subsetneq Z_k\subsetneq\cdots$ where each $Z_i$ is a union of at most $t$ subspaces, the chain terminates in at most $td$ steps.
Now we can claim that
$$
y=\bigcap_{n_d \mathrm{~is~a~"Star"~node.}} n_d.
$$
\end{proof}
We provide the characterization of ${\lozenge}\square \phi$,
\begin{lem}\label{characterization1}
For $\phi=\vee_{i=1}^t p_i$ with $p_i\in AP$, let $x$ be the maximal invariant of $\phi$ defined in Lemma \ref{invariant}. ${\lozenge}\square \phi$ if and only if $\sigma_0\vDash \psi$ where $\psi$ is the maximal extension of $x$ defined in Lemma \ref{extension}.
\end{lem}
\begin{proof}
Let $\eta$ denote the set that ${\lozenge}\square \phi$ if and only if $\sigma_0\vDash \eta$. For any $\alpha$, if $\sigma_0\vDash \E_{\alpha}(\eta)$, ${\lozenge}\square \phi$ is valid. This implies $\E_{\alpha}(\eta)\subseteq \eta$, therefore, $\vee_{\alpha\in Act}\E_{\alpha}(\eta)\subseteq \eta$.
We define the sequence $Z_0=\eta$, $Z_1=\vee_{\alpha\in Act}\E_{\alpha}(Z_0)$, $\cdots$, $Z_k=\vee_{\alpha\in Act}\E_{\alpha}(Z_{k-1})$,$\cdots$. We can verify $Z_0\supseteq Z_1\supseteq \cdots\supseteq Z_k\supseteq \cdots$. Let $y=\bigcap_{k=0}^{\infty} Z_k$, we have $y=\vee_{\alpha\in Act}\E_{\alpha}(y)$. According to ${\lozenge}\square \phi$, we know that $y\subseteq \phi$. This means $y\subseteq x$ according to the fact $x$ is the maximal invariant of $\phi$ in Lemma \ref{invariant}.
By defining an increasing sequence $V_0=x$, $V_1=\vee_{\alpha\in Act}\E_{\alpha}^{-1}(V_0)$, $\cdots$, $V_k=\vee_{\alpha\in Act}\E_{\alpha}^{-1}(V_{k-1})$,$\cdots$ we have
$\eta\subseteq \vee_{k=0}^{\infty} V_k=\psi$.

On the other hand, if $\sigma_0\vDash \psi$, for each $\omega$ path $\alpha_1\alpha_2\cdots$, let state $\sigma_{k}=\E_{\alpha_k}(\sigma_{k-1})$.
According to the proof of Lemma \ref{extension}, there exists $k_0$ such that $\rho_k\vDash x$ for all $k>k_0$.
Therefore ${\lozenge}\square \phi$ is valid.
\end{proof}
The following theorem naturally follows from Lemma \ref{invariant}, Lemma \ref{extension} and Lemma \ref{characterization1}.
\begin{theorem}
${\lozenge}\square \phi$ is decidable for $\phi=\vee_{i=1}^t p_i$ and $|Act|>1$ with $p_i\in AP$.
\end{theorem}
The following lemma is crucial in proving the decidability of $\square\lozenge \phi$.
\begin{lem}\label{characterization2}
For $\phi=\vee_{i=1}^s p_i$ with $p_i\in AP$, we can construct $x\subseteq \H$ in finite steps, such that
\begin{enumerate}
\item $\vee_{\alpha\in Act}\E_{\alpha}(x)= x$.
\item For any simple loop (distinct elements of the loop at different locations), $$x_{j_1}\xrightarrow{\E_{\alpha_1}} x_{j_2}\xrightarrow{\E_{\alpha_2}}\cdots x_{j_{k}}\xrightarrow{\E_{\alpha_k}}x_{j_1}$$ where $x_{j_1},x_{j_2},\cdots,x_{j_{k}}\in AP$, $x_{j_1},x_{j_2},\cdots,x_{j_{k}}\subseteq x$ and $\E_{\alpha_{1}}(x_{j_1})= x_{j_2}$,$\cdots$, $\E_{\alpha_{k}}(x_{j_k})= x_{j_1}$, there exists an element $x_{j_r}$ of the loop and $1\leq i\leq s$ such that $x_{j_r}\subseteq p_i$.
\item For any $x' \subseteq \H$ satisfies the first two conditions, we always have $x'\subseteq x$.
\end{enumerate}
Moreover, such $x$ is also a finite union of subspaces.
\end{lem}
The proof of this theorem depends on Lemma \ref{p2}.
\begin{proof}
We can construct $x$ using the following procedure.
\begin{itemize}
\item $l_0$: Set $x=\H$, goto $l_1$;
\item $l_1$: If Condition 1 and 2 of Lemma \ref{characterization2} are satisfied, return $x$; otherwise, goto $l_2$;
\item $l_2$: Run procedure of Lemma \ref{invariant}, goto $l_3$;
\item $l_3$: Run procedure of Lemma \ref{p2}, goto $l_1$;
\end{itemize}
For any finite union of subspaces as input, Lemma \ref{invariant} and Lemma \ref{p2} terminate within finite time. $l_1,l_2,l_3$ can only be executed a finite number of times according to Lemma \ref{uss} because the intermediate data is always a finite union of subspaces in decreasing order.

For any $x' \subseteq \H$ satisfies the first two conditions, $x'\subseteq x$ is always valid during the execution of the above procedure. Therefore, $x'\subseteq x$ is valid for the final $x$.
\end{proof}

In the proof of Lemma \ref{characterization2}, if Condition (1) of Lemma \ref{characterization2} is satisfied, but Condition (2) is not satisfied, we need the following lemma.
\begin{lem}\label{p2}
Given $\phi=\vee_{i=1}^t p_i$ and $x=\vee_{i=1}^{l} x_i$ with $p_i,x_i\in AP$, if $\vee_{\alpha\in Act}\E_{\alpha}(x)=x$ and there is a simple loop (distinct elements of the loop at different locations), $$x_{j_1}\xrightarrow{\E_{\alpha_1}} x_{j_2}\xrightarrow{\E_{\alpha_2}}\cdots x_{j_{k}}\xrightarrow{\E_{\alpha_k}}x_{j_1}$$ where $x_{j_1},x_{j_2},\cdots,x_{j_{k}}\in AP$ and $\E_{\alpha_{1}}(x_{j_1})= x_{j_2}$,$\cdots$, $\E_{\alpha_{k}}(x_{j_k})= x_{j_1}$, such that $x_{j_r}\not\subseteq \phi$ for any $j_r$.
We can find $z$ as a finite union of proper subspaces in $x_{j_1}$ such that any $x'\subseteq x$ satisfies
\begin{itemize}
\item There are infinite many $i$ such that $\sigma_i\vDash p$, for any $\omega$ sequence $\sigma_0\sigma_1\cdots$ with $\sigma_i=\E_{\beta_i}(\sigma_{i-1})$ for any $\beta_i\in Act$ with $\sigma_1\vDash x'$.
\end{itemize}
must satisfy
$$
x'\subseteq \vee_{i=1,i\neq j_1}^l x_i\vee z.
$$
\end{lem}
\begin{proof}
For any $\sigma_0\vDash x_{j_1}\bigcap x'$, we consider sequence $\omega=\sigma_0\sigma_1\cdots\sigma_r\cdots$ such that $\sigma_i=\E_{\alpha_i}(\sigma_i)$. We can divide $\omega$ into $k$ subsequence, $\omega_1=\sigma_1\sigma_{k+1}\cdots$, $\cdots$, $\omega_{k}=\sigma_{k}\sigma_{2k}\cdots$.
According to the properties $x'$ satisfies, we know that there exist $1\leq s\leq t$ and $1\leq r\leq k$ such that the sequence $\omega_r$ contains infinite state which satisfies $p_s$. In other words, there exists infinite $n$ such that $\sigma_{nk+r}=\F_r^n(\xi)\vDash p_s$ where $\xi=\E_{\alpha_r}\circ\E_{\alpha_{r-1}}\circ\cdots\circ\E_{\alpha_1}(\sigma_0)$, and $\F_{r}=\E_{\alpha_r}\circ\cdots\circ\E_{\alpha_1}\circ\E_{\alpha_k}\circ\cdots\E_{\alpha_{r+1}}$. Let $a_i=\tr(\F_r^n(\xi)p_s^{\perp})$.
Then the linear recurrent series $a_i$ has infinite zero points. Theorem \ref{SML} implies that, there is a finite $b,c$ such that
$a_{bu+c}=0$ for all $u\in \mathbb{N}$. Moreover, one can compute a $g\geq b$ which depends on $\F_{r}$ only. Therefore, we can choose $b=g!$ and let $c$ range over all integer less than $b$. Interestingly, for fixed $c$, we only need to verify $a_{bu+c}=0$ for $0\leq u\leq d^2+1$ since $a_{bu+c}$ is a linear recurrent series with degree $d^2$. In other words, $\sigma_0$ should satisfy that for some $1\leq s\leq t$, $1\leq r\leq k$, $1\leq c\leq b$, the following is true for any $0\leq u\leq d^2+1$,
$$
\F_{r}^{bu+c}\circ\E_{\alpha_r}\circ\E_{\alpha_{r-1}}\circ\cdots\circ\E_{\alpha_1}(\sigma_0)\vDash p_s.
$$
For each $s,r,c$, the constrain is equivalent to a subspace $z_{s,r,c}$. Moreover, $x_{j_1}$ is not a subspace of any $z_{s,r,c}$ according to the condition of $x$.
Let $z=\vee_{s,r,c}(x_{j_1}\bigcap z_{s,r,c})$ and it satisfies the requirement.
\end{proof}
Now we show the following
\begin{theorem}\label{tilxx}
For $p_i\in AP$ and $\phi=\vee_{i=1}^t p_i$, $\square{\lozenge} \phi$ is decidable.
\end{theorem}
\begin{proof}
Let $\psi'$ denote the set that ${\lozenge}\square \phi$ if and only if $\sigma_0\vDash \psi'$.

We first compute $x=\vee_{i=1}^l x_i$ of $\phi$ as illustrated in Lemma \ref{characterization2}, then compute $\psi$ as the maximal extension of $x$ in Lemma \ref{extension}.
We prove that $\psi'=\psi$.

To show $\psi\subseteq \psi'$, we choose $\sigma_0\vDash \psi$. For each $\omega$ path $\alpha_1\alpha_2\cdots$, we let state $\sigma_{k}=\E_{\alpha_k}(\sigma_{k-1})$.
According to the proof of Lemma \ref{extension}, there exists $k_0$ such that $\sigma_k\vDash x$ for any $k>k_0$. For any $k>k_0$, the state sequence $\sigma_{k}\sigma_{k+1}\cdots\sigma_{k+s}$ satisfies $\sigma_{k+i}\vDash x_{j_i}$ for some $1\leq j_i\leq l$. Moreover, there is a simple loop for sufficiently large $s$ because $l$ is finite. Invoking the Condition (2) of Lemma \ref{characterization2}, $\sigma_{k+u}\vDash x$ for some $1\leq u\leq s$. $\square{\lozenge} \phi$ is valid. Therefore, $\psi\subseteq \psi'$.

To show $\psi\supseteq \psi'$, we observe that for any $\alpha$, if $\sigma_0\vDash \E_{\alpha}(\psi')$, $\square{\lozenge} \phi$ is valid. This implies $\E_{\alpha}(\psi')\subseteq \psi'$.
Therefore, $\vee_{\alpha\in Act}\E_{\alpha}(\psi')\subseteq \psi'$.
We define the sequence $Z_0=\psi'$, $Z_1=\vee_{\alpha\in Act}\E_{\alpha}(Z_0)$, $\cdots$, $Z_k=\vee_{\alpha\in Act}\E_{\alpha}(Z_{k-1})$,$\cdots$. We can verify $Z_0\supseteq Z_1\supseteq \cdots\supseteq Z_k\supseteq \cdots$. Let $x'=\bigcap_{k=0}^{\infty} Z_k$, we have $x'=\vee_{\alpha\in Act}\E_{\alpha}(x')$.
Moreover, for any $\sigma_0\vDash x$, there are infinite many $i$ such that $\sigma_i\vDash p$, for any $\omega$ sequence $\sigma_0\sigma_1\sigma_2\cdots$ with $\sigma_i=\E_{\beta_i}(\sigma_{i-1})$ for any $\beta_i\in Act$ with $\sigma_1\vDash x'$. By the proof of Lemma \ref{characterization2} and Lemma \ref{p2}, $x'\subseteq x$.
By defining an increasing sequence $V_0=x$, $V_1=\vee_{\alpha\in Act}\E_{\alpha}^{-1}(V_0)$, $\cdots$, $V_k=\vee_{\alpha\in Act}\E_{\alpha}(V_{k-1})$,$\cdots$ we have
$\psi'\subseteq \vee_{k=0}^{\infty} V_k=\psi$.

Therefore $\psi=\psi'$. This implies $\square{\lozenge} \phi$ is decidable.
\end{proof}
\subsection{Until and Almost Surely Until}
According to the results of $\lozenge$. we have
\begin{fact}
\begin{itemize}
For $|Act|=1$, $q\in AP$ and $\phi=\vee_{i=1}^t p_i$ with $p_i\in AP$,
\item $\phi \mathbf{U} q$ is decidable if $\lozenge q$ is decidable.
\item $\phi \tilde{\mathbf{U}} q$ is decidable if $\tilde{\lozenge} q$ is decidable.
\end{itemize}
\end{fact}
\begin{proof}
We first verify $\square \phi$. If this is valid, we only need to verify $\lozenge q$ or $\tilde{\lozenge} q$. Otherwise, we find $n$ such that $\E^{n}(\sigma_0)\nvDash \phi$, we only need to verify $\E^{n}(\sigma_0)\nvDash q$.
\end{proof}
\begin{fact}
For $|Act|=1$, $q\in AP$ and $\phi=\vee_{i=1}^t p_i$ with $p_i\in AP$,
 $\square \phi \tilde{\mathbf{U}} q$ is decidable.
\end{fact}
\begin{proof}
 $\square \phi \tilde{\mathbf{U}} q$ iff $\square \phi$ and $\square\tilde{\lozenge} q$.
The rest follows from Theorem \ref{inv} and Fact \ref{tilde1}.
\end{proof}
We observe the following, for general $|Act|>1$.
\begin{theorem}
$\square \phi \mathbf{U} \psi$ is decidable for $\phi=\vee_{i=1}^t p_i$ and $\psi=\vee_{j=1}^s p_j$.
\end{theorem}
\begin{proof}
$\square \phi \mathbf{U} \psi$ iff $\square \phi$ and $\square\lozenge \psi$.
Verification of $\square \phi$ and $\square\lozenge \psi$ follows from Theorem \ref{inv} and Theorem \ref{tilxx}.
\end{proof}

\section{Discussion and Conclusion}
In this paper, we introduce a quantum temporal logic for quantum concurrent programs. Our quantum temporal
logic supports hierarchical specification and reasoning in a simple, natural way.
By proving a quantum B\"{o}hm-Jacopini theorem of deterministic quantum concurrent programs, we provide a simple and new insight of quantum programs written in the widely studied Q-While language.
We study the decidability of basic QTL formulae and solve the open question in \cite{LY14}.

We are fully aware of the fact that, in this paper, we have only touched upon the topic of quantum temporal logic.
There are several important directions for future work. First, it is interesting to
further develop the decidability of hierarchical specification in quantum temporal logic.
Secondly, we would like to introduce linear-time properties including fairness and liveness into the quantum concurrent program model. A framework for reasoning about the linear-time properties of concurrent unitary program is given in \cite{YLYF14}.

Thirdly, we expect our quantum temporal logic will be useful in designing quantum computing systems. We believe it is possible because our quantum temporal logic
describs a complex quantum system through a hierarchy of levels of abstraction, starting from a high-level specification and ending with implementation in some programming language.
\section{Acknowledgement}                            
We thank Prof Mingsheng Ying's help discussion on B\"{o}hm-Jacopini theorem and change the name of functional program.
We thank Prof Yaroslav Shitov for telling us the new bound \cite{Shitov19}.

This work is supported by DE180100156.

\bibliographystyle{abbrv}
\bibliography{main}

\end{document}